\documentclass[acmsmall,nonacm]{acmart}

\AtBeginDocument{%
  }

\usepackage{amsmath,amsfonts,amsthm,thmtools}
\usepackage{todonotes}
\usepackage[utf8]{inputenc}
\usepackage{graphicx}
\usepackage{macros}
\usepackage{comment}
\usepackage{multicol}
\usepackage{multirow}
\usepackage[capitalize,nameinlink]{cleveref}
\newcommand{\crefdefpart}[2]{%
  \hyperref[#2]{\namecref{#1}~\labelcref*{#1}~\!(\ref*{#2})}%
}
\usepackage{enumerate,enumitem}

\newtheorem{claim}{Claim}

\usepackage{mathtools}

\newcommand{\db}{\mathbf{db}}

\newcommand{\hyperGraphRepr}[1]{\mathcal H_{#1}}

\newcommand{\absproblemdef}[3]{%
\begin{center}
\begin{tabular}{lp{9.5cm}}\toprule
\textsf{\bfseries Problem:}& #1 \\\midrule
\textsf{\bfseries Input:}& #2.\\
\textsf{\bfseries Question:}& #3?\\\bottomrule
\end{tabular}
\end{center}
}

\usetikzlibrary{patterns.meta}

\renewcommand\footnotetextcopyrightpermission[1]{} %
\begin{document}

\title{Parameterised Complexity of Consistent Query Answering via Graph Representations}

\author{Teemu Hankala}
\email{teemu.hankala@helsinki.fi}
\affiliation{%
\institution{University of Helsinki}
  \city{Helsinki}
  \country{Finland}
}

\author{Miika Hannula}
\affiliation{%
  \institution{University of Tartu}
  \city{Tartu}
  \country{Estonia}}
\email{hannula@ut.ee}

\author{Yasir Mahmood}
\affiliation{%
  \institution{Universität Paderborn}
  \city{Paderborn}
  \country{Germany}
}

\author{Arne Meier}
\affiliation{%
 \institution{Leibniz Universität Hannover}
 \city{Hannover}
 \country{Germany}}

\renewcommand{\shortauthors}{Hankala et al.}

\begin{abstract}
We study consistent query answering via different graph representations. First, we introduce solution-conflict hypergraphs in which nodes represent facts and edges represent either conflicts or query solutions. Considering a monotonic query and a set of antimonotonic constraints, we present an explicit algorithm for counting the number of repairs satisfying the query based on a tree decomposition of the solution-conflict hypergraph. The algorithm not only provides fixed-parameter tractability results for data complexity over expressive query and constraint classes, but also introduces a novel and potentially implementable approach to repair counting.
Second, we consider the Gaifman graphs arising from MSO descriptions of consistent query answering. Using a generalization of Courcelle's theorem, we then present fixed-parameter tractability results for combined complexity over expressive query and constraint classes.
\end{abstract}



\maketitle
\section{Introduction}

Database inconsistency can arise in a number of different situations. For instance, manual errors at data entry stage, updates on duplicated data, hardware errors, software bugs, concurrency issues, and schema changes can all lead to inconsistent information being stored. Additionally, integrating data from multiple sources can introduce inconsistencies, even if each source database is individually consistent.
To manage inconsistent data in a principled way, Arenas et al. \cite{ArenasBC99} introduced the notions of a repair and consistent query answering. 
In informal terms, a \emph{repair} of a potentially inconsistent database $\db$ is a consistent database that is ``minimally different'' from $\db$. Typically the distance between two databases is measured using symmetric difference, but also other notions are possible.

A central computational problem related to database repairs is that of \emph{consistent query answering} (CQA). The input to this problem comprises a Boolean query $q$, a constraint set $\Sigma$, and a database $\db$, and the question is whether or not $q$ is true on all repairs of $\db$ over $\Sigma$. %
Over the past decade, significant research has been dedicated to charting the landscape of the computational complexity of CQA. A key factor in achieving success has been the development of graph representations to capture the interaction between the query and the constraints. In particular, the concept of an \emph{attack graph}, introduced by Wijsen~\cite{Wijsen12}, eventually led to a complete classification of the data complexity of CQA with respect to self-join-free conjunctive queries and primary-key constraints~\cite{KoutrisW17,KoutrisW21}. Later, attack graphs have been used together with \emph{dependency graphs} \cite{FaginKMP05} when addressing an extension to unary foreign keys \cite{HannulaW22}.

Another important facet of consistent query answering is its counting version. This problem is to compute the number of repairs of $\db$ satisfying a Boolean query $q$. 
Maslowski and Wijsen~\cite{MaslowskiW13} established a dichotomy between the membership in $\FP$ and $\sharpP$-completeness under Turing reductions for self-join-free conjunctive queries and primary keys.
Later this result was  %
extended to functional dependencies ({\FD}s) \cite{CalauttiLPS22a}.
To address the intractable side of repair counting, 
one line of research has been to study randomised approximation schemes under different repair semantics~\cite{CalauttiCP21,CalauttiLPS24a,CalauttiLPS22}. 
For instance, Calautti~et~al.~\cite{CalauttiLPS24a} showed that computing the percentage of operational repairs satisfying a query admits an FPRAS in combined complexity when attention is restricted to self-join-free conjunctive queries of bounded generalised hypertreewidth.

Another line of research on database repairs is to investigate the graph structure of constraint violation.
The \emph{conflict hypergraphs}, studied in several works~\cite{ChomickiMS04,CHOMICKI200590,ArenasBCHRS03,StaworkoCM12}, provide a robust abstraction for anti-monotonic constraint sets. 
In a conflict hypergraph, vertices correspond to the facts of a database, while hyperedges represent the subset-minimal collections of facts violating a constraint. 
When the constraints are {\FD}s, each hyperedge reduces to size two, resulting in a \emph{conflict graph}. 
Previously, Chomicki et al.  \cite{ChomickiMS04} presented an approach to CQA  using conflict hypergraphs. 
Their approach however is restricted to quantifier-free grounded queries for which the data complexity of CQA is already in $\Ptime$ under denial constraints.
Recently, for instance, Kimelfeld~et~al.~\cite{KimelfeldLP20} considered the enumeration of preferred repairs at the general level of
conflict graphs. Furthermore, Livshits~et~al.~\cite{LivshitsKW21} has shown that computing the number of repairs over an {\FD} set $\Sigma$ %
is in $\FP$ exactly when the {conflict graph} of each database over $\Sigma$ is $P_4$-free.

In parameterised complexity~\cite{DBLP:series/txcs/DowneyF13}, a problem's complexity is examined by additional structural components, known as parameters. 
This perspective arises from the observation that many real-world problem instances have structural features that can remain manageable or constant in size as inputs grow.
A key aim in parameterised complexity is to identify parameters that reflect practical relevance and grow slowly or remain fixed in size. 
A problem is termed \emph{fixed-parameter tractable} (FPT) if there exist an algorithm and a computable function $f$ such that for each instance $(x,k)$ the problem is solved in time $f(k)\cdot|x|^{O(1)}$. 
Problems that are in FPT are seen to be efficiently solvable in the parameterised sense.
In this paper, we explore CQA from the vantage point of graph representations and parameterised complexity, an approach that, to the best of the authors' knowledge, has been rarely considered in this context.
Lopatenko and Bertossi~\cite{LopatenkoB07} present fixed-parameter tractability results, though their focus is somewhat specific, concentrating on the incremental complexity of consistent query answering with cardinality repairs.

A very prominent parameter in the context of graphs is the treewidth~\cite{DBLP:journals/jct/RobertsonS84,bb72,halin76}. 
Intuitively, the treewidth of a graph measures how tree-like the graph is. 
Several classical NP-complete decision problems become FPT (e.g., Independent Set, Vertex Cover) when parameterised by the treewidth of the input graph. 
For non-graph problems, one has to define specific problem related graph representations of the instances reaching FPT once their treewidth is bounded. 
Regarding propositional satisfiability, it has been shown that several different such graph representations of formulas (e.g., primal graph~\cite{DBLP:conf/vldb/Yannakakis81,DBLP:journals/ai/GottlobSS02}, incidence graph and dual graph~\cite{DBLP:journals/dam/CourcelleMR01}) yield FPT algorithms. 
Concerning graph representations, Courcelle's Theorem~\cite{DBLP:journals/eatcs/CourcelleE12,DBLP:journals/iandc/Courcelle90,DBLP:journals/algorithmica/BoriePT92} is one of the most commonly known tools, stating that every problem expressible in monadic second-order logic (MSO) can be solved in linear time on structures whose graph representations are of bounded treewidth.

\paragraph*{Our Contributions}
We investigate CQA using graph representations that capture the combined interaction of $\Sigma$, $q$, and $\db$, offering a unified perspective to the three-partite structure of the problem. 
More specifically, we identify two primary graph approaches to database repairs and consistent query answering: one considers
the minimal solution to the query $q$ 
and another that focuses on partial solutions for pairs of query atoms.
As a result, both graphs empasise on the interplay between 
$q$ and the database $\db$ while also modelling the conflicts due to constraints in $\Sigma$. 
    The first representation extends the conflict hypergraph with hyperedges corresponding to minimal query solutions, and the second one is motivated by Courcelle's theorem.
    For the first graph representation (called the \emph{solution-conflict hypergraph}), we present a direct algorithm that leads to fixed-parameter tractability results for the data complexity of the counting version of CQA. 
    However, it does not serve as a solution to combined complexity due to an exponential blow-up in the size of the graph. 
    In contrast, for the second graph representation, we utilise Courcelle's theorem obtaining combined complexity results.
    We emphasise that the two graph representations and their resulting treewidth measures are incomparable to each other (Theorem~\ref{thm:tw-different}). 
    Moreover, the obtained results and techniques extend to expressive monotonic queries (such as unions of conjunctive queries with inequality) and anti-monotonic constraints (such as denial constraints).
Table~\ref{tab:results} depicts an overview of our results.

\begin{table}[]
\centering
\resizebox{\textwidth}{!}{
    \begin{tabular}{cccccc}
    \toprule 
    $\rightarrow$ ICs  & \multicolumn{2}{c}{data complexity} & \multicolumn{2}{c}{combined complexity} \\
    $\downarrow$ Parameters  & PKs / FDs & DCs & PKs / FDs & DCs \\
    \midrule 
    $\tw_\calC/\blocksize$ 
    & $\para\co\NP$-complete$^{{\text{ C}\ref{cor:tw-conf}}}$ 
    & $\para\co\NP$-hard$^{{\text{ C}\ref{cor:tw-conf}}}$ 
    & $\para\co\NP$-hard$^{{\text{ C}\ref{cor:tw-conf}}}$ 
    & $\para\co\NP$-hard$^{{\text{ C}\ref{cor:tw-conf}}}$
     \\
    $\atomsq$ 
    & $\para\co\NP$-complete$^{\text{~\cite{FuxmanM07}}}$
    & $\para\co\NP$-hard$^{\text{~\cite{FuxmanM07}}}$
    & $\para\co\NP$-hard$^{\text{~\cite{FuxmanM07}}}$ 
    & $\para\PiP$-complete$^{\text{~\cite{decidability03}}}$\\
    $\twH$/$\twG$ &
    $\FPT^{\footnotesize{\text{ C}\ref{cor:FPT-solution-conflict}/\text{ C}\ref{cor:gaifman-dc-fds}}}$ 
    & $\FPT^{\footnotesize{\text{ C}\ref{cor:FPT-solution-conflict}/\text{ T}\ref{thm:gaifman-dcs}}}$ 
     & -- 
     & -- \\
    $\twG+\atomsq$
    & $\FPT^{\footnotesize{\text{ T}\ref{cor:gaifman-dc-fds}}}$ 
    & $\FPT^{\footnotesize{\text{ T}\ref{thm:gaifman-dcs}}}$
    & $\FPT^{\footnotesize{\text{ T}\ref{thm:gaifman-cc-fds}}}$ 
    & $\FPT^{\footnotesize{\text{ T}\ref{thm:gaifman-dcs}}}$ \\
    \bottomrule
    \end{tabular}
  }
    \caption{Parameterised complexity results for $\cqa$ w.r.t.\ $\bucq^{\neq}$-queries. The lower bounds apply already to $\bcq$-queries, and
    the upper bounds in the top two rows are straightforward.
    The lower bounds regarding $\atomsq$ follow trivially from the existing ones, those are listed here only for completeness. The membership results for DCs assume the sets of constraints to be fixed. 
    Each result is followed by a pointer to the result containing the proof. 
    }
    \label{tab:results}
\end{table}

\section{Preliminaries}
Given a positive integer $n$, we let $[n]\coloneqq\{1,\dots,n\}$. 
We assume disjoint countably infinite sets of \emph{variables}, \emph{constants}, and \emph{relation names}.
A \emph{term} is either a variable or a constant. %
A mapping from a set of variables to a set of constants is called an \emph{assignment}. An assignment is extended to be the identity on the constants.%

\subsection{Databases and queries}
A \emph{(database) schema} is a finite set $\sch=\{R_1, \dots ,R_n\}$ of {relation names}. Each relation name $R$ is associated with an \emph{arity} $\ar(R)\geq 1$,
and we sometimes write $R[k]$ instead of $R$ to emphasise that $R$ is a relation name with arity $k$. %
For terms %
$t_{1},\dots,t_{k}$, the expression $R(t_1, \dots ,t_k)$ is called a \emph{(relational) atom}.
A relational atom composed only of constants is called a \emph{fact}. An expression of the form $t=t'$, where $t$ and $t'$ are terms, is called an \emph{(equality) atom}.

Fix a database schema $\sch$.
A \emph{database over $\sch$} is a finite set $\db$ of $R$-facts for $R \in \sch$. The \emph{active domain of $\db$}, denoted $\adom{\db}$, is the set of constants appearing in $\db$.
A \emph{first-order logic formula $\phi$ over $\sch$}
is an expression built from relational atoms $R(\vec{t})$, $R\in \sch$, and equality atoms %
 using logical connectives $\neg$, $\lor$, $\land$ and quantifiers $\exists$, $\forall$.
The \emph{set of free variables of $\phi$}, denoted $\FV(\phi)$, 
is defined in the usual manner, and $\phi$ is called a \emph{sentence} if %
$\FV(\phi)=\emptyset$.   The set of constants appearing in $\phi$ is denoted $\adom{\phi}$.
Given an assignment $s\colon\FV(\phi) \to \adom{\db}\cup\adom{\phi}$, we say that  \emph{$\db$ satisfies $\phi$ under $s$}, written $(\db,s)\models \phi$, if the first-order structure $\db$ with domain $\adom{\db}\cup\adom{\phi}$ and  assignment  $s$ satisfies $\phi$ according to the usual Tarskian semantics. 
For a sentence $\phi$, we say that \emph{$\db$ satisfies $\phi$} and write $\db \models \phi$ if $\db$ satisfies $\phi$ under the empty assignment.

A \emph{Boolean query $q$ over $\sch$} can be considered to be a set of databases over $\sch$.  
If $\db\in q$, we say that it \emph{satisfies} $q$, writing $\db \models q$.
The query $q$ is \emph{monotonic} (resp. \emph{anti-monotonic}) if $\db \subseteq \db'$ (resp. $\db \supseteq \db'$) and $\db \models q$ imply $\db \models q'$.
A first-order logic sentence is interpreted as a Boolean query via the set of databases that satisfy it.
A \emph{Boolean conjunctive query} %
($\bcq$)
is a first-order sentence $\phi$ of the form
\begin{equation}\label{eq:bcq}
\exists \vec{y} (R_1(\vec{t_1}) \land \dots \land R_m(\vec{t}_m)), %
\end{equation}
where for $i\in [m]$, 
$R_i(\vec{t}_i)$ is a relational atom over a sequence $\vec{t}_i$ of %
terms. A \emph{self-join-free Boolean conjunctive query} ($\sjfbcq$) is a $\bcq$ in which no two distinct atoms share a relation name. %
A \emph{Boolean conjunctive query with inequality} ($\bcq^{\neq}$) is obtained from \cref{eq:bcq} by extending its quantifier-free part with negated equality atoms $\neg t= u$ (in short, \emph{inequality atoms} $t \neq u$), where $t$ is a variable occurring in some relational atom and $u$ is a constant or a variable occurring in some relational atom.
A \emph{Boolean union of conjunctive queries with inequality} ($\bucq^{\neq}$) is a first-order sentence $\phi$ of the form
\(
\psi_1 \lor \dots \lor \psi_\ell
\),
 where $\psi_i$ is a  $\bcq^{\neq}$, for $i\in [\ell]$.

A \emph{constraint} over $\sch$ is simply a Boolean query over $\sch$.
A \emph{denial constraint} (DC) is a first-order logic sentence of the form
\(
\forall \vec{x} \neg(\theta_1 \land \dots \land \theta_n), 
\)
where $\theta_i$ is either a relational atom, an equality atom, or a negated equality atom, for $i\in [n]$. A \emph{functional dependency} (FD), is an expression of the form $R\colon U\to V$, where $U,V \subseteq [\ar(R)]$. It represents the
first-order logic sentence 
\[
\forall \vec{x}\vec{y} \left(\left(R(\vec{x}) \land R(\vec{y}) \land \bigwedge_{i\in U}x_i =y_i\right)\to \bigwedge_{i\in V}x_i =y_i\right).
\]
Each {\FD} can also be expressed as a DC. 
A \emph{key} is an expression of the form $R:U$, for $U \subseteq [\ar(R)]$, representing the {\FD} $R\colon U\to [\ar(R)]$. If a relation $R$ is associated with a unique key $R:U$, we call this key the \emph{primary key} (PK) of $R$. In particular, a set of primary keys does not contain two different keys over a shared relation name.

Let $\Sigma$ be a fixed set of constraints. 
We say that $\db$ \emph{satisfies} $\Sigma$, written $\db \models \Sigma$, if $\db$ satisfies each constraint in $\Sigma$.
A {database} $\db$ is called \emph{consistent (over $\Sigma$)} %
if it satisfies $\Sigma$, and otherwise \emph{inconsistent}. %
Let us denote by $\oplus$ the symmetric difference operation, and
let $\db$ be a (potentially inconsistent) database. 
A \emph{repair of $\db$ (over $\Sigma$)} is any consistent database $\rep$ for which any other database $\sep$ such that $\sep \oplus \db \subsetneq\rep\oplus \db$ is inconsistent.
From this it follows that if the constraints of $\Sigma$ are anti-monotonic,
then the repair $\rep$ must be a subset of $\db$. Such repairs are often referred to as \emph{subset-repairs} in the literature.

We consider the following decision problem.
\absproblemdef{$\cqa$}{%
A database $\db$, set of constraints $\Sigma$, and a Boolean query $q$}{Is $q$ true on every repair of $\db$ over $\Sigma$}
$\ncqa$ is defined analogously and counts the number of repairs satisfying $q$.
If $\mathcal{C}$ is a class of constraints, and $\mathcal{Q}$ a class of Boolean queries,
we denote by $\cqa(\mathcal{C},\mathcal{Q})$ the restriction of $\cqa$ to input instances $(\db,\Sigma,q)$ where $q$ belongs to $\mathcal{Q}$ and each constraint of $\Sigma$ belongs to $\mathcal{C}$.
We write $\cqa(\Sigma,q)$ for the restriction of $\cqa$ to input instances $(\db,\Sigma,q)$; i.e., only the database $\db$ is given, whereas $\Sigma$ and $q$ are fixed. Note that this problem can be understood as the \emph{data complexity} of CQA, whereas $\cqa(\mathcal{C},\mathcal{Q})$ is understood as the \emph{combined complexity} of CQA (with respect to $\mathcal{C}$ and $\mathcal{Q}$).
The restrictions for $\ncqa$ are defined analogously.

\subsection{Parameterised Complexity} 
We give a brief exposition of parameterised complexity theory and refer to the textbook of Downey and Fellows~\cite{DBLP:series/txcs/DowneyF13} for details.
A \emph{parameterised problem (PP)} $\Pi$ is a subset of $\Sigma^*\times\mathbb N$, where $\Sigma$ is an alphabet.
For an instance $(x,k)\in\Sigma^*\times\mathbb N$, $k$ is called the \emph{parameter value}.
A PP $\Pi$ is \emph{fixed-parameter tractable} (short: $\FPT$) if there exists a deterministic algorithm deciding $\Pi$ in time $f(k)\cdot|x|^{O(1)}$ for every input $(x,k)$, where $f$ is a computable function.
Similarly, a \emph{parameterised counting problem} $F \colon \Sigma^* \times \mathbb N \to \mathbb N$
is fixed-parameter tractable, if there is a deterministic algorithm that for
some computable function $f$ computes $F(x, k)$ in time
$f(k) \cdot |x|^{O(1)}$ for each $(x, k) \in \Sigma^* \times \mathbb N$
\cite{DBLP:journals/siamcomp/FlumG04}.

\begin{definition}\label{def:fptreduction}
	Let $\Sigma$ and $\Delta$ be two alphabets.
	 A PP $\Pi\subseteq\Sigma^*\times\mathbb{N}$ \emph{$\FPT$-reduces} to a PP $\Theta\subseteq\Delta^*\times\mathbb N$, in symbols $\Pi\fptreduction\Theta$, if %
	(i) there is an $\FPT$-computable function $f$, such that, for all $(x,k)\in\Sigma^*\times\mathbb N$: $(x,k)\in \Pi\Leftrightarrow f(x,k)\in \Theta$,
	(ii) there exists a computable function $g\colon\mathbb N\to\mathbb N$ such that for all $(x,k)\in\Sigma^*\times\mathbb N$ and $f(x,k)=(y,\ell)$: $\ell\leq g(k)$.
\end{definition}

The problems $\Pi$ and $\Theta$ are $\FPT$-\emph{equivalent} if both $\Pi\fptreduction\Theta$ and $\Theta\fptreduction\Pi$ are true.
We also use higher complexity classes via the concept of \emph{precomputation on the parameter}.
\begin{definition}
	Let $\mathcal C$ be any complexity class.
	Then $\para\mathcal C$ is the class of all PPs $\Pi\subseteq\Sigma^*\times\mathbb N$ such that there exist a computable function $\pi\colon\mathbb N\to\Delta^*$ and a language $L\in\mathcal C$ with $L\subseteq\Sigma^*\times\Delta^*$ such that for all $(x,k)\in\Sigma^*\times\mathbb N$ we have that $(x,k)\in \Pi \Leftrightarrow (x,\pi(k))\in L$.
\end{definition}
Observe that $\para\Ptime=\FPT$ is true.
For a constant $c\in\mathbb N$ and a PP $\Pi\subseteq\Sigma^*\times\mathbb N$, the \emph{$c$-slice of $\Pi$}, written as $\Pi_c$, is defined as $\Pi_c\coloneqq\{\,(x,k)\in\Sigma^*\times\mathbb N\mid k=c\,\}$.
In our setting, showing $\Pi\in\para\mathcal C$, it suffices to show $\Pi_c\in\mathcal C$ for every $c\in\mathbb{N}$.
Moreover, in order to prove that a PP $\Pi$ is $\para\mathcal C$-hard for some complexity class $\mathcal C$, it is enough to prove that $\Pi_c$ is $\mathcal C$-hard for some $c\in \mathbb N$.

\subsection{Hypergraphs and treewidth}
A \emph{hypergraph} is a pair $\hgsym=(V,E)$ consisting of a set $V$ of nodes and a set $E$ of subsets of $V$, called \emph{hyperedges}. A hyperedge of size $1$ or $2$ is called an \emph{edge}.
A simple \emph{graph} (without loops) can thus be viewed as a particular kind of hypergraph, in which each hyperedge is an edge (of size $2$).
\begin{definition}[Treewidth]\label{def-tw}
A \emph{tree decomposition} of a hypergraph $\hgsym=(V,E)$ is a tree $T=(B,E_T)$, where the vertex set $B\subseteq\mathcal P(V)$ is a collection of \emph{bags} and $E_T$ is a set of edges as follows:
\begin{enumerate}
	\item $\bigcup_{b\in B}b=V$,
	\item\label{it:subsumed} for every $e\in E$ there is a bag $b\in B$ with $e\subseteq b$, and 
	\item for all $v\in V$ the subtree of $T$ induced by the bags containing $v$ is connected.
\end{enumerate} 
The \emph{width} of the tree decomposition $T$
is the size of the largest bag decreased by one: $\max_{b\in B}|b|-1$.
The \emph{treewidth} of $\hgsym$, denoted $\tw(\calH)$, is the minimum width over all tree decompositions of $\hgsym$.
\end{definition}
The \emph{primal graph} of a hypergraph $\hgsym$ is the (simple) graph $G$ that has the same node set as $\hgsym$, and an edge between each pair of nodes that belong to a common hyperedge. 
The treewidth of $\mathcal{H}$ equals the treewidth of its primal graph.

The conflict hypergraphs, mentioned already in the introduction, are defined as follows.
\begin{definition}[Conflict hypergraph]\label{def:conf-graph}
    Let $\db$ be a database and $\Sigma$ a set of anti-monotonic constraints.
    The \emph{conflict hypergraph} $\hyperGraphRepr{\db,\Sigma}$ of $(\db,\Sigma)$
    is defined as the hypergraph $\hgsym=(V,E)$, where 
    \begin{enumerate}
        \item the set of nodes $V$ is $\db$,
        \item\label{it:conflict} a subset $e\subseteq \db$ that is minimal w.r.t.\ falsifying the constraints forms a hyperedge in $E$, i.e.,
        \begin{enumerate}
            \item $e \not\models \Sigma$ and 
            \item for each $e' \subsetneq e$ we have that $e'\models \Sigma$.
        \end{enumerate}
    \end{enumerate}
    We refer to a \emph{conflict graph} when the conflict hypergraph is a graph.
    In general, constructing the conflict hypergraph has exponential time complexity in $\size{\db}+\size{\Sigma}$.
\end{definition}

\section{Warm-up: Conflict Graphs}
We first establish two simple lower bounds for conjunctive queries and primary keys.

Fix a set of primary keys $\Sigma$ over a database $\db$. We say that two facts $R(\vec{a}),R(\vec{b})\in \db$ are \emph{key-equal} if they agree on their primary-key positions; that is, if $R:U$ is the primary key of $R$ in $\Sigma$, then $a_i=b_i$ for $i\in U$.
A \emph{block} is any subset-maximal set of key-equal facts in $\db$.
\begin{proposition}[{\cite[Lemma~5]{FuxmanM07}}]\label{prop:coNP}
    Consider a set of PKs $\Sigma= \{R_1\colon \{1\} , R_2 \colon \{1\} \}$ and a $\sjfbcq$ $q=\exists x y z  (R_1(x,y)\land R_2(z,y))$.
    Then $\cqa{(\Sigma, q)}$ is $\co\NP$-hard.
\end{proposition}
The above result is obtained via a reduction from (the complement of) satisfiability of monotone 3CNF formulae, resulting in a database  $\db$ having blocks of size $3$.
Consequently, the $\co\NP$-hardness holds already with respect to a fixed maximum block size.
Furthermore, beside the cliques formed by the blocks, the conflict graph of $\db$ %
has no other edges.
As a result, its treewidth %
also remains fixed.
This allows us to deduce the following completeness result.

\begin{corollary}\label{cor:tw-conf}	
    Let $\Sigma$ be a set of PK{s} and $q$ a $\sjfbcq$.
    Then, $\cqa(\Sigma,q)$ is  $\para\co\NP$-complete when parameterised by (1) the maximum block size in $\db$, or (2) the treewidth of the conflict graph of $(\db, \Sigma)$.
\end{corollary}
Note that one can lift the upper bound to any query language whose data complexity is in $\P$.

\section{Solution-conflict graphs}\label{sec:solution-conflict}
Database repairs are known to correspond to maximal independent sets in conflict hypergraphs~\cite{CHOMICKI200590,StaworkoCM12}.
While the conflict hypergraph captures the violations of constraints, it does not include information about the query. To apply conflict hypergraphs in consistent query answering, they need to be enriched with information about query solutions. This brings us to the concept of a solution-conflict hypergraph.
\begin{definition}[Solution-conflict hypergraph]\label{def:sol-conf-graph}
    Let $\db$ be a database, $\Sigma$ a set of anti-monotonic constraints, and $q$ a monotonic Boolean query.
    The \emph{solution-conflict hypergraph} $\hyperGraphRepr{\calI}$ of the triple $\calI=(\db,\Sigma,q)$ extends the conflict hypergraph of $(\db,\Sigma)$
    by adding the following hyperedges:
     \setcounter{enumi}{2} %
    \begin{enumerate}
        \item[(3)]%
        a subset $e\subseteq \db$ that is minimal w.r.t.\ satisfying the query forms a hyperedge in $E$, i.e.,
        \begin{enumerate}
            \item $e \models q$ and 
            \item for each $e' \subsetneq e$ we have that $e'\not\models q$.
        \end{enumerate}
        
    \end{enumerate}
    We refer to (hyper)edges formed according to \crefdefpart{def:conf-graph}{it:conflict} as \emph{conflict-(hyper)edges} and to (hyper)edges formed according to (3) as \emph{solution-(hyper)edges}.
    We then write $E_{\rm s}$ for the set of solution-(hyper)edges and $E_{\rm c}$ for the set of conflict-(hyper)edges of $E$.
\end{definition}    

First, we observe that %
in some specific cases the maximum independent sets of the solution-conflict hypergraph and the database repairs falsifying the query have a straightforward connection. 
A set of nodes %
in a hypergraph is called \emph{independent} if it does not contain any hyperedge. 
A \emph{maximum independent set} is an independent set of maximal cardinality. 

\begin{theorem}\label{thm:graph1-arbitrary}
Let $\db$ be a database having $k$ blocks, $\Sigma$ be a set of PKs, and $q$ be a Boolean monotonic query.
Then, there exists a repair $\rep$ of $\db$ over $\Sigma$ such that $\rep \not\models q$
    if and only if the solution-conflict graph of $(\db,\Sigma,q)$ has an independent set of size $k$. %
\end{theorem}
\begin{proof}
Suppose the solution-conflict hypergraph of 
$(\db, \Sigma, q)$ has an independent set $\rep$ of size $k$. Then, $\rep$ contains exactly one fact from each block, identifying it as a repair for $\db$ over primary key constraints. 
Furthermore, $\rep$ neither contains any conflict-edge (hence $\rep\models \Sigma$) nor any solution-hyperedge (thus $\rep\not \models q$). %

Conversely, suppose no $k$-element subset of the solution-conflict hypergraph of $(\db, \Sigma, q)$ is an independent set, and let $\rep$ be an arbitrary repair of $\db$. Since $\rep$ selects exactly one fact from each block, its size is $k$. Since $\rep$ is a repair, it cannot contain any conflict-edge. 
Thus by the hypothesis $\rep$ contains a solution edge. 
As a result, by monotonicity of $q$, we obtain $\rep \models q$. %
\end{proof}
The previous theorem entails that CQA over primary keys and monotonic queries can be reduced to the problem of computing the size of the maximum independent set in the solution-conflict hypergraph. 
Namely, some repair of a database $\db$ falsifies a query $q$ if and only if the size of maximum independent set for the solution-conflict hypergraph equals the size of the maximum independent set for the conflict hypergraph, which in turn equals the number of blocks in $\db$.
This approach, however, no longer applies when dealing with constraint sets more expressive than primary keys (or when more than one key per relation). 
The next example show the case of two keys. %
Note that an independent set is called \emph{maximal} if none of its strict supersets are independent. 
\begin{example}\label{ex:running}
     Consider a database $\db = \{R(a,b), R(c,b), R(c,d), R(e,d), R(e,f)\}$ and a constraint set $\Sigma$ consisting of two keys $R:\{1\}$ and $R:\{2\}$.
Then, a repair does not have a fixed size by virtue of $\{R(a,b), R(c,d), R(e,f)\}$, $\{R(c,b),R(e,f)\}$, and $\{R(c,b), R(e,d)\}$ being all possible repairs. 
The notion of blocks does not seem applicable to measuring the maximum independent set size in the conflict graph.
A repair of $\db$ may falsify a query even when maximum independent set sizes differ between solution-conflict and conflict hypergraphs. 
For the grounded query $q=R(a,b)$, the sets $\{R(c,d), R(e,f)\}$, $\{R(c,b),R(e,f)\}$, and $\{R(c,b), R(e,d)\}$ are maximal for the solution-conflict hypergraph, implying that the maximum independent set has a smaller cardinality in this graph than in the conflict hypergraph. 
Yet it is not the case that all repairs satisfy $q$, as witnessed by the repairs $\{R(c,b),R(e,f)\}$ and $\{R(c,b), R(e,d)\}$.\hfill$\blacktriangleleft$
\end{example}
For the remainder of this section, we fix a monotonic query $q$ and a set of anti-monotonic constraints $\Sigma$. 
Our aim is to construct a direct algorithm for $\ncqa(\Sigma,q)$, leveraging a tree decomposition of the associated solution-conflict hypergraph. 
To achieve this, we propose an algorithm that is executed twice. 
The algorithm is related to the problem of counting the number of maximal independent sets (\#IS), which is $\sharpP$-complete already for bipartite graphs \cite{ProvanB83}.  
The first execution computes the number of maximal independent sets in the conflict hypergraph, effectively solving \#IS on hypergraphs. 
The second execution determines the number of maximal independent sets in the conflict hypergraph that falsify $q$. 
The final result for $\ncqa(\Sigma,q)$ is then obtained by subtracting the second count from the first. 

Let us first consider a na\"ive approach to compute the total number of repairs: multiply the number of repairs for each bag. The next example demonstrates why this approach fails.
\begin{example}
    Consider $\db$ and $\Sigma$ from \cref{ex:running}. The conflict graph is a path 
     \[
     R(a,b) \text{ --- } 
     R(c,b) \text{ --- } 
     R(c,d) \text{ --- }  
     R(e,d) \text{ --- } R(e,f).
     \]
     It has the following tree decomposition of treewidth $1$:
          \begin{equation}\label{eq:path}
           \{R(a,b), R(c,b)\} \text{ --- }  
           \{R(c,b), R(c,d)\} \text{ --- } 
           \{R(c,d), R(e,d)\} \text{ --- } 
           \{R(e,d),R(e,f)\}.
          \end{equation}
     Each bag has exactly two singleton repairs. Multiplying the number of repairs for the bags yields $16$. However, since the bags overlap, this number overestimates the total number of repairs of $\db$. Indeed, $\db$ has only $4$ repairs, the ones listed in:
     \begin{multicols}{2}
        \begin{enumerate}
            \item $\{R(a,b), R(c,d), R(e,f)\}$,
            \item $\{R(c,b), R(e,d)\}$,
            \item $\{R(a,b),R(e,d)\}$,
            \item $\{R(c,b),R(e,f)\}$.
        \end{enumerate}
     \end{multicols}
\end{example}
Hence, the algorithm requires more sophisticated bookkeeping to manage the overlaps between different bags. 
Before moving on, let us define some terminology. %
Let $\db$ be a database and $\rep,\sep$ sets such that $\rep \subseteq \sep \subseteq \db$. We say that $\rep$ is a \emph{max-repair of $\sep$ (in $\db$)} if $\rep$ is a repair of $\sep$, and there is no $\sep'$, $\sep \subsetneq \sep' \subseteq \db$, such that $\rep$ is a repair of $\sep'$. Since we assume that the underlying constraint set $\Sigma$ is anti-monotonic, a set $\rep$ is a repair of $\sep_0 \cup \sep_1$ whenever it is a repair of $\sep_0$ and $\sep_1$ individually.
Thus $\rep$ is a max-repair of some unique set. 
 Given a consistent subset $\rep$ of $\db$, we hence write
 \[
\maxrep{\rep}{\db} \coloneqq \bigcup \{\sep \subseteq \db \mid \rep \text{ is a repair of }\sep\}.
 \]
 In particular, $\maxrep{\rep}{\db}$ is the unique set $\sep$ such that $\rep$ is a max-repair of $\sep$ in $\db$. The following simple lemma provides an alternative characterization for $\maxrep{\rep}{\db}$.
\begin{lemma}\label{lem:desc}
    Let $\Sigma$ be a set of anti-monotonic constraints.
   Let $\rep \subseteq \db$ be consistent over $\Sigma$. Then 
   \[
   \maxrep{\rep}{\db}  = \rep \cup \{A \in \db \mid \rep \cup \{A\}\not\models \Sigma\}.
   \]
 \end{lemma}
 \begin{proof}
     ``$\subseteq$": Let $A \in \sep \setminus \rep$, where $\rep$ is a repair of $\sep$. Then by definition $\rep\cup \{A\}$ is inconsistent, meaning that $\rep \cup \{A\}\not\models \Sigma$.
     ``$\supseteq$": Let $A\in \db$ be a fact such that $\rep \cup \{A\}\not\models \Sigma$. Then $\rep$ is a repair of $\rep \cup \{A\}$, and thus $A \in \maxrep{\rep}{\db}$.
 \end{proof}
 If both $\rep$ and $\db$ are consistent, we note that $\rep$ is a max-repair of itself. 
  Consider then a tree decomposition $T=(B,E)$ of the solution-conflict hypergraph of $(\db,\Sigma,q)$. %
  We fix some root $r$ of $T$.
  For each bag $b \in B$, let $\child{b}$ be the set of bags that are children of $b$, and let $\sub{b}$ be the set of nodes included in the subtree rooted at $b$. We extend this notation to sets of bags $S$ by writing $\sub{S}\coloneqq \bigcup_{b \in S} \sub {b}$. Furthermore, if $C \subseteq \child{b}$, we then write $\subtwo{b}{C}$ for $\{b\} \cup \sub{C}$. To avoid repetitive use of the union symbol, we also write $\usub{b}\coloneqq\bigcup \sub{b}$, $\usub{S}\coloneqq \bigcup \sub{S}$, and $\usubtwo{b}{S}\coloneqq \bigcup \subtwo{b}{S}$, for the aforementioned sets of bags.

 Let $b\in B$, $C \subseteq \child{b}$ and $r \subseteq s \subseteq b$. Let $c$ be the least element of $C$ according to some ordering.
 Define 
 \begin{align*}
\f(r,s,b,C) &= \begin{cases}
h(r,s,b,C)
& \text{ if $C=\emptyset$,}\\
       \sum_{\substack{s' \cup  s'' = s \cap  c\\ r \cap c\subseteq s' \cap s''}}
    \f(r,(s \setminus c)\cup s',b, C \setminus \{c\}) \cdot \g(r \cap c, s'',b,c)
    & \text{ otherwise,}
\end{cases}\\
\text{where }h(r,s,b,C)&=
    \begin{cases}
    1 & \text{ if $r$ is a max-repair of $s$ in $b$ such that $r \models q$,}\\
    0 & \text{ otherwise.}
\end{cases} 
\end{align*}
Moreover, given $b\in B$, $c \in \child{b}$, and $r \subseteq s \subseteq b \cap c$,  define
\begin{equation*}
  \g(r,s,b,c) = \sum_{
r' \subseteq c \setminus b
}
\f(r\cup r',s \cup (c\setminus b),c, \child{c}). 
\end{equation*}
We will use these functions to compute the number of repairs of $\db$ satisfying $q$. To show that the functions work as intended, we next define a set of extensions of a set $r$ into a max-repair.
Given a subset of bags $S\subseteq E$ and sets $r \subseteq s \subseteq d \subseteq \usub{S}$, we define $\smax{r}{s}{d}{S}$ as the set
\begin{equation}\label{eq:eset}
    \{r' \subseteq \usub{S} \setminus d \mid \maxrep{r\cup r'}{\usub{S}}=s\cup  (\usub{S}\setminus d), \forall b \in S \forall e \in E_{\rm s}: e \not\subseteq (r \cup r') \cap b\}.
\end{equation}
In words, $\smax{r}{s}{d}{S}$ consists of all subsets $r'$ of $\usub{S}\setminus d$, where $r\cup r'$ is a max-repair of $s \cup (\usub{S}\setminus d)$, and its intersection with any bag of $S$ does not contain any solution-edge of $T$ (see \cref{fig:eset}).
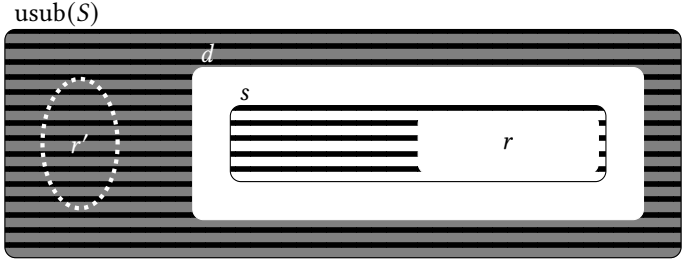
\begin{figure}
\centering

\begin{tikzpicture}%

	\filldraw[gray,rounded corners] (0,0) -- (9,0) -- (9,3) -- (0,3) -- cycle;
	\filldraw[rounded corners,pattern={Lines[
                  distance=2mm,
                  angle=45,
                  line width=0.7mm
                 ]},opacity=.1,
        pattern color=black] (0,0) -- (9,0) -- (9,3) -- (0,3) -- cycle;
	\node at (.7,3.2) {$\usub{S}$};
	
	\draw[ultra thick,dotted,white] (1,1.5) ellipse (.5cm and .85cm);
	\node at (1,1.5) {\textcolor{white}{$r'$}};
	
	\filldraw[gray,rounded corners] (2.5,0.5) -- (8.5,0.5) -- (8.5,2.5) -- (2.5,2.5) -- cycle;
	\filldraw[white,rounded corners,opacity=.5] (2.5,0.5) -- (8.5,0.5) -- (8.5,2.5) -- (2.5,2.5) -- cycle;
	\node at (2.7,2.7) {\textcolor{white}{$d$}};
	
	\filldraw[white,rounded corners,opacity=.5] (3,1) -- (8,1) -- (8,2) -- (3,2) -- cycle;
	\filldraw[rounded corners,pattern={Lines[
                  distance=2mm,
                  angle=45,
                  line width=0.7mm
                 ]},opacity=.1,
        pattern color=black] (3,1) -- (8,1) -- (8,2) -- (3,2) -- cycle;
	\node at (3.2,2.15) {$s$};
	
	\filldraw[white,rounded corners,opacity=.8] (5.5,1.1) -- (7.9,1.1) -- (7.9,1.9) -- (5.5,1.9) -- cycle;
	\node at (6.7,1.5) {$r$};
	
\end{tikzpicture}

\caption{Given a subset of bags $S\subseteq E$ and sets $r \subseteq s \subseteq d \subseteq \usub{S}$, $\smax{r}{s}{d}{S}$ consists of all $r'\subseteq \usub{S}\setminus d$ such that $r \cup r'$ form a max-repair of the hatched area $s\cup (\usub{S}\setminus d)$, and $e \not\subseteq(r\cup r')\cap b$ for any $b\in S,e\in E_s$. \label{fig:eset}}
\end{figure}%
The key technical lemma is stated below, with its proof provided in \Cref{app:solconfhyp}.
\begin{restatable}{lemma}{biglemma}\label{lem:induction}The following equalities are true:
        \begin{enumerate}
        \item\label{it:one} $\f(r,s,b,C) = \nmax{r}{s}{b}{\subtwo{b}{C}}$,
\item\label{it:three} $\g(r,s,b,c)=
        \nmax{r}{s}{b\cap c}{\sub{c}}$.
    \end{enumerate}
\end{restatable}

Since every minimal solution is contained within a single bag of the tree decomposition, the query's satisfaction can be verified in a modular fashion, as the next simple lemma shows.
\begin{lemma}\label{lem:inbag}
    Consider a tree decomposition $T=(B,E)$
    of the solution-conflict hypergraph of $(\db,\Sigma,q)$, where $q$ is a monotonic Boolean query.
    Let $\rep \subseteq \db$. Then
    $\rep \models q \text{ if and only if }\rep \cap b \models q\text{ for some }b\in B.$
\end{lemma}
\begin{proof}
    ``$\Leftarrow$": Follows by monotonicity of $q$. ``$\Rightarrow$": Suppose $\rep \models q$, and let $e\in E$ be some solution-edge witnessing this. In particular, $e \subseteq \rep$ and $e\models q$. By \crefdefpart{def-tw}{it:subsumed} we find some bag $b\in B$ such that $e \subseteq b$. By monotonicity of $q$, $\rep \cap b \models q$. 
\end{proof}

 The previous lemmas thus lead us to the following characterisation.
\begin{lemma}
Let $\Sigma$ be a set of anti-monotonic constraints, $\db$ a database, and $q$ a monotonic Boolean query.
    Consider a tree decomposition of the solution-conflict hypergraph of $(\db,\Sigma,q)$ that is rooted at $a$. 
    Then $\sum_{r \subseteq a}\f(r,a,a,\child{a})$ is the number of repairs $\rep$ of $\db$ such that $\rep \not\models q$. 
\end{lemma}
\begin{proof}
    Denote by $T=(B,E)$ the tree decomposition. 
    Note that $\rep$ is a repair of $\db$ if and only if it is a max-repair of $\db$ in $\db$. 
    Thus, by \crefdefpart{lem:induction}{it:one} and \cref{lem:inbag} we obtain
    \begin{align*}
    &\f(r,a,a,\child{a})=        \nmax{r}{a}{a}{\sub{a}}=\nmax{r}{a}{a}{B} =|\{r' \subseteq \usub{B} \setminus a \mid \maxrep{r\cup r'}{\usub{B}}\\
    &=a\cup  (\usub{B}\setminus a)\text{ and }\forall b \in B: (r \cup r') \cap b\not\models q\}|\\
    &=    |\{r' \subseteq \db \setminus a \mid {r\cup r'}\text{ is a repair of }\db\text{ and }r \cup r'\not\models q\}|.
    \end{align*}
    Consequently, $\sum_{r \subseteq a}\f(r,a,a,\child{a}) = |\{\rep \subseteq \db \mid \rep\text{ is a repair of }\db\text{ and }\rep\not\models q\}|.$
\end{proof}
\Cref{alg:falsify} now computes the number of repairs that falsify a query by computing the necessary values of functions $f$ and $g$ for each bag of the tree decomposition. The correctness of the algorithm follows from the previous result. Note that in the algorithm, given a set $S$, we write $\vec{S}$ for any sequence listing the elements of $S$ (without repetitions). For a sequence $\vec{a}$ and an element $b$,
we also denote by $\vec{a}b$ the sequence formed by appending $\vec{a}$ with $b$.

\begin{algorithm}
\caption{\textsc{NumberFalsify}$(\db,\Sigma,q)$\label{alg:falsify}}
    \KwIn{A database $\db$, a set of anti-monotonic constraints $\Sigma$, a monotonic Boolean query $q$\;}
    \KwRequire{
    A tree decomposition $T=(B,E)$ of $\hyperGraphRepr{\db,\Sigma,q}$ rooted at $a$\,
    }
    \KwOut{The number of repairs of $\db$ over $\Sigma$ falsifying $q$}
    Initialise both $\mathsf{f}$ and $\mathsf{g}$ as the empty functions\;
    \For{$b\in B$}{
                \For{$r,s \subseteq b$ such that $r \subseteq s$}{
      \If{$r$ is a max-repair of $s$ in $b$ such that it contains no solution edge of $\hyperGraphRepr{\db,\Sigma,q}$}{
           $\mathsf{f}(r, s, b, ()) \gets 1$\;
       }
       \lElse{
           $\mathsf{f}(r, s, b, ()) \gets 0$
    }
    }
    Mark $(b,())$\;
    }
    \While{exists $b\in B$ and $c \in \child{b}$ such that $(b,c)$ is not marked but $(c,\vec{\child{c}})$ is marked}{
    \lFor{$r,s \subseteq b\cap c$ such that $r \subseteq s$}{
     $\mathsf{g}(r,s,b,c) = \sum_{
r' \subseteq c \setminus b
}
\mathsf{f}(r\cup r',s \cup (c\setminus b),c, \vec{\child{c}}) $
}
Mark $(b,c)$\;
    }
    \While{exists $b,c\in B$ and an initial segment $\vec{d}c$ of  $\vec{\child{b}}$ such that $(b,\vec{d}c)$ is not marked but $(b,\vec{d})$ and $(b,c)$ are marked
    }{
    \For{$r,s \subseteq b$ such that $r \subseteq s$}{
     $\mathsf{f}(r,s,b,\vec{d}c) = \sum_{\substack{s' \cup  s'' = s \cap  c\\ r \cap c\subseteq s' \cap s''}}
    \mathsf{f}(r,(s \setminus c)\cup s',b, \vec{d}) \cdot \mathsf{g}(r \cap c, s'',b,c)$
}
Mark $(b,\vec{d}c)$\;
    }
    \Return $\sum_{r \subseteq a}\mathsf{f}(r,a,a,\vec{\child{a}})$\;
\end{algorithm}

Let us examine the algorithm through the lens of our running example.
\begin{example}
    Consider the set of keys $\Sigma$ from Ex.~\ref{ex:running}, and a database 
    $\db'=\{R(a,b), R(c,b), R(c,d)\}.$
     It is clear that $\db'$ has two repairs over $\Sigma$, $\{R(a,b),  R(c,d)\}$ and $\{R(c,b)\}$. Consider the tree decomposition
    $\{R(a,b), R(c,b)\} - \{R(c,b), R(c,d)\}.$ 
     Taking $\{R(a,b), R(c,b)\}$ as the root, we demonstrate that $\textsc{NumberFalsify}(\db,\Sigma,q_{\bot})$ outputs $2$ as required. Let us denote $b=\{R(a,b), R(c,b)\}$ and $c= \{R(c,b),R(c,d)\} $. 
     
     During the first for-loop the algorithm sets $\mathsf{f}(r, s, b_0, ())=1$ if and only if $r$ is a singleton and $s=b_0$, for all bags $b_0$.
    The first while-loop then computes 
        \begin{align*}
           \mathsf{g}(\emptyset,\emptyset,b,c)=& \sum_{
r' \subseteq c \setminus b
}
\mathsf{f}( r', c\setminus b,c, ()) = \mathsf{f}( \emptyset, c\setminus b,c, ()) + \mathsf{f}( c\setminus b, c\setminus b,c, ())=0+0=0,\\
\mathsf{g}(\emptyset,b \cap c,b,c)=& \sum_{
r' \subseteq c \setminus b
}
\mathsf{f}( r', c,c, ()) = \mathsf{f}( \emptyset, c,c, ())+ \mathsf{f}( c\setminus b, c,c, ())=0+1=1,\\
\mathsf{g}(b\cap c,b\cap c,b,c) = &\sum_{
r' \subseteq c \setminus b
}
\mathsf{f}((b \cap c)\cup r',c,c, ()) = \mathsf{f}(b \cap c,c,c,())+ \mathsf{f}(c,c,c, ())=1+0=1.\\
        \end{align*}   
After this, the second while-loop computes
\begin{align*}
      \mathsf{f}(b,b,b,(c)) = &       \sum_{s' \cup  s'' = b \cap  c
    }\overbrace{\mathsf{f}(b,(b \setminus c)\cup s',b, ())}^0 \cdot \mathsf{g}(b \cap c, s'',b,c)=0,\\
      \mathsf{f}(\emptyset,b,b,(c)) = &       \sum_{s' \cup  s'' = b \cap  c
    }\overbrace{\mathsf{f}(\emptyset,(b \setminus c)\cup s',b, ())}^0 \cdot \mathsf{g}(\emptyset, s'',b,c)=0,\\
         \mathsf{f}(b\setminus c,b,b,(c)) = &       \sum_{s' \cup  s'' = b \cap  c
    }\mathsf{f}(b\setminus c,(b \setminus c)\cup s',b, ()) \cdot \mathsf{g}(\emptyset, s'',b,c)\\
    = &\, \mathsf{f}(b \setminus c,b \setminus c,b, ()) \cdot \mathsf{g}(b \cap c, b \cap c,b,c)
     + \mathsf{f}(b\setminus c,b,b, ()) \cdot \mathsf{g}(\emptyset, \emptyset,b,c)\\
     &\,+ \mathsf{f}(b\setminus c,b,b, ()) \cdot \mathsf{g}(\emptyset, b \cap c,b,c)\\
                      =&\, 0\cdot 1 + 1 \cdot 0 + 1 \cdot 1=1,\\
         \mathsf{f}(b \cap c,b,b,(c)) = &       \sum_{s' \cup  s'' = b \cap  c 
    }\mathsf{f}(b \cap c,(b \setminus c)\cup s',b, ()) \cdot \mathsf{g}(b \cap c, s'',b,c)\\
    =&\, \mathsf{f}(b\setminus c,b \setminus c,b, ()) \cdot \mathsf{g}(b \cap c, b\cap c,b,c)
        +\, \mathsf{f}(b \cap c,b,b, ()) \cdot \mathsf{g}(b \cap c, b \cap c,b,c)\\
                  =&\, 0\cdot 1  + 1 \cdot 1=1.
    \end{align*}
Finally, as intended, the algorithm returns
\[
\sum_{r \subseteq a}\mathsf{f}(r,b,b,(c))= \mathsf{f}(\emptyset,b,b,(c))+\mathsf{f}(b \setminus c,b,b,(c))+\mathsf{f}(b \cap c,b,b,(c)) + \mathsf{f}(b,b,b,(c)) = 0+1+1+0=2.
\]
\end{example}

We can now present the main result of the section.
\begin{theorem}\label{lem:time}
Fix a set of anti-monotonic constraints $\Sigma$ and a monotonic Boolean query $q$. 
Given a database $\db$ and a tree decomposition $T=(B,E)$
    of the solution-conflict hypergraph of $(\db,\Sigma,q)$, the number of repairs $\rep$ of $\db$ over $\Sigma$ satisfying $q$ can be computed in time $|E|\cdot O( 3^{2k})$, where $k$ is the maximum size of a bag in $B$.
\end{theorem}
\begin{proof}
Clearly, $T$ is also a tree decomposition of the conflict hypergraph of $(\db,\Sigma)$. Such a hypergraph can be viewed as the
solution-conflict hypergraph of $(\db,\Sigma,q_\bot)$, where $q_\bot\coloneqq \bot$, i.e., $\db \not\models q_\bot$ for all $\db$.
In particular, $\textsc{NumberFalsify}(\db,\Sigma,q_{\bot})$ computes the number of repairs of $\db$ over $\Sigma$. Thus, the number of repairs $\rep$ of $\db$ over $\Sigma$ satisfying $q$ is obtained as
\[
\textsc{NumberFalsify}(\db,\Sigma,q_{\bot}) - \textsc{NumberFalsify}(\db,\Sigma,q).
\]
To analyse the time complexity, note first that
there are $3^k$ ways to find $r,s\subseteq b$ such that $r\subseteq s$. Since $\Sigma$ and $q$ are fixed, the first for-loop can be computed in time $|B|\cdot O( k^c\cdot 3^{k})$ for some constant $c$. The first while-loop computes in time $|E|\cdot O( 3^{k}\cdot 2^k)$, and the second in time $|E|\cdot O( 3^{2k})$; in particular, the subscripts of the sums hide for-loops. The time complexity of the second while-loop dominates the others, and hence the statement of the lemma follows.
\end{proof}

   We obtain as a corollary that $\ncqa(\Sigma,q)$ is fixed-parameter tractable in the treewidth of the solution-conflict hypergraph with respect to many common query and constraint classes.
   Specifically, with respect to data complexity the solution-conflict hypergraphs exhibit no fundamental distinction  between conjunctive queries and their unions (as such queries are monotonic), or between keys, functional dependencies, and denial constraints (as such constraints are anti-monotonic).
   The next result demonstrates one example corollary. %
We write $\twH(\calI)$ for the treewidth of the solution-conflict hypergraph of  $\calI=(\db,\Sigma,q)$. %
\begin{corollary}\label{cor:FPT-solution-conflict}
Let $\Sigma$ be a set of DCs and $q$ an $\bucq^{\neq}$.
Then, the problems $\ncqa(\Sigma,q)$ and $\ncqa(\FD,q)$ are fixed-parameter tractable when parameterised by $\twH$.
\end{corollary}
\begin{proof}
We consider the case of $\ncqa(\Sigma,q)$, where the only input is the database $\db$. The case of $\ncqa(\FD,q)$, where the {\FD} set $\Sigma$ is part of the input, is analogous.

First note that $q$ has a fixed bound on the size of its minimal solutions, and similarly %
$\Sigma$ has a fixed bound on the size of its minimal conflicts. Thus the solution-conflict hypergraph $\hgsym$ can be constructed from $\db$ in time $poly(\size{\db})$. Given a graph $G$, one can construct a tree decomposition~$\mathcal{T}$ of~$G$ of treewidth at most~$4 (k+1)$ in time~$2^{O(k)}\cdot poly(\size{G})$ \cite[Thm. 7.18]{DBLP:books/sp/CyganFKLMPPS15}. (Note that ordinary graphs are here interchangeable with hypergraphs, since any tree decomposition of a hypergraph is also a tree decomposition of its primal graph, and vice versa).
Finally, by \cref{lem:time} %
the output of $\ncqa(\Sigma,q)$ can be computed in time $\size{\mathcal{T}}\cdot 2^{O(k)}$. Composing the findings yields a time bound
$poly(\size{\db}) \cdot 2^{O(k)}$, which leads to the statement of the theorem. 
\end{proof}

\section{%
Gaifman Graphs}
In this section, we examine the tractability of the CQA problem using Gaifman graphs derived from MSO descriptions of the problem. 
Using Courcelle's theorem~\cite{DBLP:journals/eatcs/CourcelleE12,DBLP:journals/iandc/Courcelle90,DBLP:journals/algorithmica/BoriePT92}, we then obtain \FPT results parameterised by the treewidth of the Gaifman graph. 
Recall that the \emph{Gaifman graph} of a first-order structure \( \calA \) over a vocabulary $\tau$ is a graph whose vertex set consists of the elements in the domain of \( \calA \), and there is an edge between any two elements \( a \) and \( b \) if and only if there is a vec \( \vec{c} \in R^\calA \), for some \( R \in \tau \), such that \( a \) and \( b \) are distinct elements occurring in \( \vec{c} \).
We note that the Gaifman graphs considered next are independent of the solution-conflict hypergraphs analysed in the previous section.

We fix some terminology first.
Let %
$\db$ be a database over a schema $\sch$, %
$\Sigma$ a set of constraints, and $q$ a $\bcq^{\neq}$ of the form
\begin{equation}\label{eq:bucq}
\exists \vec x \left(\bigwedge_{i\leq s} q_i \land \bigwedge_{s < i\leq s'} t_i \neq u_i\right),
\end{equation}
where $q_i$ is a relational atom, for $i \leq s$. %
For a list of terms $\vec t$, we write $\atomvars{\vec t}$ for the set of variables occurring in $\vec t$.
Relational atoms $q_i(\vec t)$ and $q_j(\vec u)$, $i,j\leq s$, are called \emph{$q$-linked}
if one of the following conditions holds:
\begin{enumerate}
    \item $i=j$,
    \item $\atomvars{\vec t} \cap \atomvars{\vec u}\neq \emptyset$, or
    \item there are variables $x\in \atomvars{\vec t}$ and $y\in \atomvars{\vec u}$ such that $\{t_i,u_i\}=\{x,y\}$ for some $s < i\leq s'$. 
\end{enumerate}%
Furthermore, facts $q_i(\vec a), q_j(\vec b)\in \db$ are called \emph{$q$-consistent} if
\begin{enumerate}
    \item the atoms $q_i(\vec t),q_j(\vec u)\in q$ are $q$-linked, and
    \item $h(\vec{t})= \vec{a}$, $h(\vec{u})=\vec{b}$, and $h( t_i)\neq h( u_i)$, for $s<i\leq s'$, for some assignment $h$. %
\end{enumerate}
Notice that, we use the expression $q_i(\vec{a})$ instead of $R(\vec{a})$ to emphasise that the fact is being treated as a potential solution for the atom $q_i(\vec t)$ (instead of, say, another query atom $q_j(\vec{u})$).
\begin{example}
Suppose $q= \exists xy (R(x,z)\land R(u,y) \land x \neq y)$, for variables $x,y,z,u$. Then $q_1(x,z)=R(x,z)$ and $q_2(u,y)=R(u,y)$ are $q$-linked, and $q_1(a,c)$ and $q_2(c,b)$ are $q$-consistent. However, $q_1(c,b)$ and $q_2(a,c)$ are not $q$-consistent, because no assignment $h$ can satisfy the conditions $h(x)=c$, $h(y)=c$, and $h(x)\neq h(y)$ simultaneously.
Suppose $q' = \exists x (R(x)\land x \neq c)$, for a variable $x$ and a constant $c$. Then $q'_1(x) = R(x)$ is vacuously $q'$-linked to itself. However, $q'_1(c)=R(c)$ is not $q'$-consistent (with itself) because the conditions $h(x)=c$ and $h(x)\neq h(c)$ are inconsistent for any assignment $h$ (which by definition has to fix the constants). 
\end{example}

Next we reduce $\cqa(\FD,\bucq^{\neq})$ to model checking of MSO-formulas.
We detail our analysis below, and later extend the reduction to DCs.

\subsection{The Case of Functional Dependencies}
We translate an instance $\calI=(\db,\Sigma,q)$ of $\cqa(\FD,\bucq^{\neq})$ to a pair $(\calA_\calI,\Phi_\calI)$, where $\calA_\calI$ is a first-order structure and $\Phi_\calI$ is an MSO-formula. Consider first a $\bcq^{\neq}$ $q$ of the form \cref{eq:bucq}.
Writing $L= \{(i,j)\mid i, j\leq s, q_i\text{ and } q_j \text{ are $q$-linked}\}$, %
the vocabulary of  %
interest is 
$\tau = \{\depfails[2]\}\cup\{\linked_{i,j}[2]\mid (i,j)\in L \}$.
The structure $\calA_\calI$ over $\tau$ consists of a domain $\db$ and has the following 
interpretations: %
\begin{itemize}
    \item $\depfails^\calA = \{ (f,g) \in \db \mid \{f,g\}\not\models \Sigma\}$, %
    \item $\linked_{i,j}^\calA = \{(q_i(\vec a), q_j(\vec b)) \mid  q_i(\vec a), q_j(\vec b) \in \db \text{ are $q$-consistent} \}$. %
\end{itemize}

 The relation $\depfails$ is actually the set of conflict edges arising from $\db$ and $\Sigma$ (\cref{def:conf-graph}). The relations $\linked_{i,j}$, however, do not correspond to the solution hyperedges arising from $\db$ and $q$ (\cref{def:sol-conf-graph}). On the one hand, two  facts may be joined by $\linked_{i,j}$ even if they do not belong to a common solution-hyperedge of $\hyperGraphRepr{\calI}$. On the other hand, two distinct facts from the same solution-hyperedge are only joined by some $\linked_{i,j}$ if their respective atoms share some variables or are linked by an inequality atom. 
 
 For a triple $\calI=(\db,\Sigma,q)$, %
 let $\twG(\calI)$ denote {the
treewidth of the Gaifman graph of  $\calA_{\calI}$.}
 Below, 
 we  
 establish that the treewidths obtained from $\calA_\calI$ and $\calH_\calI$ are not related. The proof can be found in \cref{sect:appendixFD}.
\begin{restatable}{theorem}{different}\label{thm:tw-different}
        There are families $(\db_n)_{n \in \mathbb{N}}$ and $(\db'_n)_{n \in \mathbb{N}}$ of databases, sets of constraints $\Sigma$ and $\Sigma'$, and $\bcq$-queries $q$ and $q'$ such that
        \begin{enumerate}
            \item  %
            $\twH(\db_n,\Sigma,q)\geq n$ and $\twG(\db_n,\Sigma,q) \leq 1$, and 
            \item %
            $\twG(\db'_n,\Sigma',q')\geq n$ and $\twG(\db'_n,\Sigma',q') \leq 1$.
        \end{enumerate}
\end{restatable}
 
 Another key distinction between $\calA_{\calI}$ and the solution-conflict hypergraph $\hyperGraphRepr{\calI}$ lies in complexity to construct them. 
 The former requires polynomial time in the size of $\calI$---even when the query $q$ is part of the input---as we only consider pairs of facts when constructing the model (i.e., not every subset of $\db$ has to be considered).
 In contrast, constructing the solution-conflict hypergraph in the same manner is not feasible unless $q$ is considered fixed. 
 Consequently, %
 it seems unlikely that %
  the fixed-parameter tractability of $\ncqa(\FD,q)$ can be extended to even $\ncqa(\FD,\bcq)$ in \cref{cor:FPT-solution-conflict}. 
 In contrast, we will establish fixed-parameter tractability results for $\calA_\calI$ with respect to combined complexity.

We now move on to the MSO-formula $\Phi_\calI$. %
Denoting by $T$ a set variable (depicting a repair), we first construct the following auxiliary formulas
\begin{enumerate}
    \item $\varphi_{\text{sat}}(T) \coloneqq  \forall x\forall y (T(x) \land T(y)\rightarrow \neg\depfails(x,y) ) $,
    \item $\varphi_{\text{$T,q$-sat}}\coloneqq \exists x_1\ldots\exists x_s (\bigwedge_{i\leq s} T(x_i) \land\bigwedge_{(i,j)\in L}  \linked_{i,j}(x_i,x_j) )$,
    \item $\varphi_{\tilde T\supsetneq T}(T, \tilde T) \coloneqq \forall x (T(x) \rightarrow \tilde T(x))\land \exists y (\neg T(y) \land \tilde T(y)) $,
    \item $\varphi_{\text{repair}}(T) \coloneqq \varphi_{\text{sat}}(T)\land \forall \tilde T (\neg \varphi_{\tilde T\supsetneq T} \lor \neg \varphi_{\text{sat}}(\tilde T))$.
\end{enumerate}
Consider now an instance $\calI=(\db,\Sigma,q)$ of $\cqa(\FD,\bucq^{\neq})$, where $q=\bigvee_{\ell\leq m}q_\ell$. 
Then, the formula $\Phi_\calI$ is formed as follows:
       $$\Phi_\calI \coloneqq \forall T\left(\varphi_{\text{repair}}(T) \to \bigvee_{\ell\leq m}  \varphi_{\text{$T,q_{\ell}$-sat}}\right).$$
Note that the construction of $\calA_\calI$ and $\Phi_\calI$ require time even polynomially bounded in the size of $\calI$. 
The proof for the correctness of our encoding, stated next, is deferred to \cref{sect:appendixFD}.
\begin{restatable}{theorem}{correct}\label{thm:mso-correct-fds}
Let $\Sigma$ be a set of FDs, $q$ a $\bucq^{\neq}$, and $\db$ a database. Then all repairs of $\db$ over $\Sigma$ satisfy $q$
    if and only if $\calA_\calI\models \Phi_\calI$.
\end{restatable}

For a fixed query $q$, the formula $\Phi_\calI$ remains fixed regardless of the database $\db$ and the FD set $\Sigma$.
This results in the following fixed-parameter tractability result due to Courcelle’s theorem~\cite{courcelle1990graph,DBLP:conf/focs/ElberfeldJT10}.
\begin{corollary}\label{cor:gaifman-dc-fds}
    Let $q$ be a $\bucq^{\neq}$. %
    The problem $\cqa(\FD,q)$ is fixed-parameter tractable when parameterised by 
 $\twG$. %
\end{corollary}

For combined complexity, our translation yields a family of MSO-formulas that depend on the query $q$.
In particular, the subformula $\varphi_{\text{$T,q$-sat}}$, which encodes the pair of query atoms in $q$, increases in size as $q$ becomes larger.
At first sight this prohibits the application of Courcelle's theorem. 
Nevertheless, following ideas from L\"uck~et~al.~\cite[Theorem~2]{luck2015parameterized}, we utilise a manifold application of Courcelle's theorem.
We next formalise this in the following, thereby making the theorem applicable.
Let $\kappa$ be a parameterisation. 
A function $f\colon \Sigma^*\rightarrow \Sigma^*$ is $\kappa$-bounded, if there is a computable function $g$ such that $|f(x)| \leq g(\kappa(x))$ for all $x\in \Sigma^\star$.
We now establish a lemma giving size bounds for a constructed MSO-formula $\Phi_\calI$ from an instance $\calI$. 
Then we can generalise \cref{cor:gaifman-dc-fds} to combined complexity by taking the number of query atoms as an additional parameter. We write $\atomsq$ for the function that maps a $\bucq^{\neq}$ $q$ to the number of atoms occurring in it.
For the proofs, the reader is referred to \cref{sect:appendixFD}.
\begin{restatable}{lemma}{lemfamily}\label{lem:family}
    Considering $\kappa=\atomsq(q)$ as the parameterisation, there is a $\FPT$-computable $\kappa$-bounded function $f$ such that $f(\calI)=\Phi_\calI$, for an instance $\calI=(\db, \Sigma, q)$ of $\cqa(\FD,\bucq^{\neq})$.
\end{restatable}

\begin{restatable}{theorem}{thmgaifman}\label{thm:gaifman-cc-fds}
    The problem $\cqa(\FD,\bucq^{\neq})$ is fixed-parameter tractable 
    when parameterised by $\atomsq%
    +\twG$. %
\end{restatable}
\begin{proof}
    Lemma~\ref{lem:family} bounds the size $|\Phi_\calI|$ by a function of the parameter $\atomsq(q)$. 
    Given an instance $\calI$, we compute $\Phi_\calI$ in $\FPT$ time. 
    Then, the model checking problem of the instance $(\calA_\calI,\Phi_\calI)$ can be decided in $\FPT$ time parameterised by $\twG$ via Courcelle's theorem. 
    Considering $\atomsq+\twG$ as the parameterisation $\kappa$, the size of the formula $\Phi_\calI$ is $\kappa$-bounded, whereby the theorem applies.
\end{proof}
Note that the only number of atoms is not automatically a parameter yielding $\FPT$. 
Specifically,  the $\para\co\NP$-hardness  for data complexity stated in \cref{prop:coNP} already holds when the query contains two atoms.

\subsection{Generalisation to Denial Constraints}
    We consider a fixed set of denial constraints for our MSO encoding.
    This is motivated by the fact that the repair checking problem for denial constraints is intractable when considering arbitrary sets of DCs. 
    Precisely, given a set $\Sigma$ of denial constraints as input, then checking whether an instance $\rep$ is a repair of an instance $\db$ w.r.t. $\Sigma$ is $\DP$-complete even if $\rep$ and $\db$ are both fixed~\cite[Lem.~6]{arming2016complexity}.
    The complexity of the repair checking drops to $\Ptime$ for a {fixed set} of denial constraints $\Sigma$; however, $\cqa(\Sigma)$ remains intractable.

For a set of DCs $\Sigma$, denote by $\atomsc(\Sigma)$ the maximum number of relational atoms appearing in any DC of $\Sigma$. Whenever a database $\db$ does not satisfy $\Sigma$, we find a subset $\db'\subseteq \db$ of size at most $\atomsc(\Sigma)$ not satisfying $\Sigma$. Thus we change the reduction in the following way.
Let $\tau'$ be a vocabulary obtained from $\tau$ by changing the arity of $\depfails$ from $2$ to $k$, where $k=\atomsc(\Sigma)$. Given an instance $I=(\db,\Sigma,q)$, let $\calB_\calI$ be the structure over $\tau'$ that is otherwise exactly as $\calA$, except that the interpretation of $\depfails$ is changed to $\{(f_1, \dots, f_k) \mid f_1,\dots ,f_n\in \db, \{f_1, \dots ,f_n\}\not\models \Sigma\}$.

Furthermore, the MSO-formula $\Psi_\calI$ is exactly as $\Phi_\calI$, except that the formula $\varphi_{\text{sat}}$ is updated as follows: 
\begin{enumerate}
    \item[{(1a)}] $\varphi_{\text{sat}}(T) 
    =\forall x_1\dots\forall x_k %
    (\bigwedge_{i\leq k} T(x_i)\rightarrow \neg \depfails(x_1,\dots,x_k))$.
\end{enumerate}

We next prove (see \cref{app:dc}) that our encoding for denial constraints is also correct. 
\begin{restatable}{theorem}{correctDC}\label{thm:mso-correct-dcs}
Let $\Sigma$ be a set of DCs, $q$ a $\bucq^{\neq}$, and $\db$ a database. Then all repairs of $\db$ over $\Sigma$ satisfy $q$ if and only if 
$\calB_\calI\models \Psi_\calI$.
\end{restatable}

Similar to the case of {\FD}s, the MSO-formula $\Psi_\calI$ in our translation remains fixed for the data complexity since the set $\Sigma$ is also fixed.
We first establish the size bounds for the MSO-formula $\Psi_\calI$ in our translation before stating the $\FPT$ results for both $\cqa(\Sigma,q)$ and $\cqa(\Sigma,\bucq^{\neq})$. The proofs are located in the appendix.

\begin{restatable}{lemma}{lemfamilydcs}\label{lem:family2}
    Fix a set of DCs $\Sigma$.
    Considering $\kappa=\atomsq(q)$ as the parameterisation, there is a $\FPT$-computable $\kappa$-bounded function $f$ such that $f(\calI)=\Psi_\calI$, for an instance $\calI=(\db, \Sigma, q)$ of $\cqa(\Sigma,\bucq^{\neq})$. 
\end{restatable}

\begin{restatable}{theorem}{gaifmandc}\label{thm:gaifman-dcs}
Let $\Sigma$ be a set of DCs and $q$ a $\bucq^{\neq}$.
    The problem $\cqa(\Sigma, q)$ is fixed-parameter tractable when parameterised by $\twG$. 
    Additionally, the problem $\cqa(\Sigma, \bucq^{\neq})$ is fixed-parameter tractable when parameterised by $\atomsq+\twG$.
    \end{restatable}
\begin{proof}
    Lemma~\ref{lem:family} bounds the size $|\Phi_\calI|$ by a function of the parameter $\atomsq(q)$. 
    Given an instance $\calI$, we compute $\Phi_\calI$ in $\FPT$ time. 
    Then, the model checking problem of the instance $(\calA_\calI,\Phi_\calI)$ can be decided in $\FPT$ time in $\twG$ by Courcelle's theorem. 
    Considering $\atomsq+\twG$ as parameterisation, the size of the formula $\Phi_\calI$ is $\kappa$-bounded, whereby the theorem applies.
\end{proof}

\paragraph*{Non-fixed DCs.}
A natural follow-up question is
whether the $\FPT$ results for $\cqa(\Sigma,q)$ and $\cqa(\Sigma,\bucq^{\neq})$ can be extended to combined complexity by considering parameters on $\Sigma$.
It appears, however, that taking $\atomsc$ as the parameter might not help reach tractability, as the exponential blow-up in constructing the structure $\calB_\calI$ cannot be limited to the parameter alone. 
Namely, the time complexity of constructing $\depfails^{\calA_\calI}$
is $O(\size{\db}^{\atomsc(\Sigma)})$.
However, bounding the maximum arity of any DC in $\Sigma$ by a constant $k$ would yield an \FPT result.
This is true since the structure $\calB_\calI$ can then be constructed in $\FPT$ time, and using the same approach as before exploiting Courcelle's theorem.

\section{Conclusion}
The present work continues the line of research dedicated to tame the intractable side of CQA. 
Different strategies for approaching the same issue include the use of SAT solvers \cite{DixitKolaitis,DixitK21,DixitK22}, binary integer programming \cite{KolaitisPT13}, logic programming (e.g., \cite{ArenasBC03,GrecoGZ01,NieuwenborghV06}), and the development of randomised approximation schemes \cite{CalauttiLP18,CalauttiLPS22a,CalauttiCP21,CalauttiLPS22,CalauttiLPS24a}.
Our approach focuses on leveraging graph parameters to achieve fixed-parameter tractability. 
Specifically, we established tractability results based on the treewidth of two novel graph representations: the solution-conflict hypergraph and the Gaifman graph derived from the MSO description of the problem. 
A notable strength of this parameterised approach lies in its robustness. 
Interestingly, our findings reveal no principled distinction between conjunctive queries and their unions, or between keys, functional dependencies, and denial constraints. 

While the results obtained highlight the theoretical robustness of our parameterised framework, the practical relevance remains a key open question. 
Courcelle's theorem, applied to the second graph representation, guarantees theoretical tractability; however, its direct applicability in practice is often limited. 
To address this gap, we designed a direct algorithm for the counting version of CQA, applicable to the first graph representation. 
Implementing and comparing this algorithm to others is next on our agenda. 
Additionally, comparing the two treewidth measures ($\twH$ and $\twG$) on real-world datasets could help identify which measure scales better in practice.
To further extend the parameterised complexity analysis, one possibility is to explore other parameterisations or different kinds of problems, e.g., enumeration complexity. 
Another possibility is to extend investigations to other constraints (e.g., foreign keys/inclusion dependencies/tgds) and queries (e.g., those with aggregates or less restricted use of negation). 
Yet another interesting direction for future work is to utilise recently developed (so-called) decomposition-guided~\cite{eiter2021treewidth, HECHER2022103651, fichte2021decomposition} reductions to obtain theoretical upper and lower runtime bounds.
These reductions would allow to establish tight runtime bounds by encoding $\cqa$ and $\#\cqa$ in quantified Boolean logic.
Exploring these directions is left for future work.

\bibliographystyle{ACM-Reference-Format}
\bibliography{main}

\appendix
\section{Solution-Conflict Hypergraphs}\label{app:solconfhyp}
In this section, we sometimes write $AB$ for the union $A \cup B$ of two sets $A$ and $B$, and consequently shorten $A\{b\}$ as $Ab$.

 \begin{lemma}\label{lem:sub}
     Let $\rep \subseteq \db$ satisfy a set  of antimonotonic constraints $\Sigma$, and let $\db' \subseteq \db$. Then $\maxrep{\rep \cap \db'}{\db'} \subseteq \maxrep{\rep}{\db} \cap \db'$.
 \end{lemma}
 \begin{proof}
     Suppose $A \in \maxrep{\rep \cap \db'}{\db'}$. By \cref{lem:desc}, $(\rep \cap \db')\cup \{A\} \not \models \Sigma$, and hence, in particular, $\rep\cup \{A\} \not \models \Sigma$. Again, by \cref{lem:desc}, $A \in \maxrep{\rep}{\db}$. The claim of the lemma then follows, as $A \in \db'$ by the hypothesis.
 \end{proof}

 \begin{lemma}\label{lem:equ}
 Let $\db$ be a database, and let $\Sigma$ be a set of antimonotonic constraints. Let $T=(B,E)$ be a tree decomposition of the conflict graph of $(\db,\Sigma,q)$. Let $\{B_1,B_2\}$ be an arbitrary partition of $B$ into two sets of bags, and let $\db_1 = \bigcup B_1$ and $\db_2 = \bigcup B_2$. Then,
     \[
\maxrep{\rep}{\db} %
= \maxrep{\rep\cap \db_1}{\db_1} \cup \maxrep{\rep \cap \db_2}{\db_2}. %
     \]
     \end{lemma}
 \begin{proof}
    Observe first that ``$\supseteq$'' is due to \cref{lem:sub}. 
    
    For ``$\subseteq$'', let $A \in \maxrep{\rep}{\db}$. Then $\rep \cup \{A\}\not\models \Sigma$. We may assume $A \notin \rep$; otherwise the set inclusion is immediate. Now, we find a subset $\rep'\subseteq \rep$ such that $\rep' \cup \{A\}$ forms a hyperedge in the conflict graph of $(\db, \Sigma,q)$. In particular, $\rep'\cup \{A\}\not\models \Sigma$.
    Moreover, $\rep' \cup \{A\}$ is included in some bag of $B$. By symmetry, assume this bag belongs to $B_1$, in which case $\rep' \cup \{A\} \subseteq \db_1$.
Then $(\rep \cap\db_1) \cup \{A\} \not\models \Sigma$ since $\rep' \cup \{A\}\subseteq (\rep \cap\db_1) \cup \{A\}$. \cref{lem:desc} then implies $A \in \maxrep{\rep\cap \db_1}{\db_1}$, concluding the proof for ``$\subseteq$''. 
 \end{proof}

\biglemma*
\begin{proof}
    The proof is by a simultaneous induction on the structure of the tree decomposition.

    \paragraph{Base steps} Consider \cref{it:one}, and suppose $b$ is a leaf bag. Since $b$ has no children, $C=\emptyset$ and the induction statement breaks down to $\f(r,s,b,\emptyset)$ being either $1$ or $0$, depending on whether $r$ is a max-repair of $s$ in $b$ satisfying $q$;
    this, indeed, is how $\f(r,s,b,\emptyset)$ is defined. %
     \cref{it:three} does not have a base step since in this case $b$ is not a leaf bag.
    
    \paragraph{Induction step (\cref{it:three})} Let $b\in B$, $c \in \child{b}$, and $r \subseteq s \subseteq b \cap c$.
    The induction hypothesis states that for any $r',s'$ such that $r' \subseteq s' \subseteq c$, 
    $\f(r',s',c)=
        \nmax{r'}{s'}{c}{\sub{c}}$.
We obtain
    \begin{align*}
  \g(r,s,b,c) &= \sum_{
r' \subseteq c \setminus b
}
\f(rr',s(c\setminus b),c)\\
&=\sum_{
r' \subseteq c \setminus b
}
\nmax{r r'}{s  (c\setminus b)}{c}{\sub{c}}\\
&=
\nmax{r}{s}{b \cap c}{\sub{c}}.
    \end{align*}

\paragraph{Induction step (\cref{it:one})} Let $b\in B$, $C \subseteq \child{b}$ and $r \subseteq s \subseteq b$.
The induction hypothesis is that for any $c\in \child{b}$ and any $r,s$ such that $r \subseteq s \subseteq b \cap c$,
\[
B(r,s,b,c)=\nmax{r}{s}{b\cap c}{\sub{c}}.
\]
We prove the induction statement by induction on the size of $C$. The base step $C=\emptyset$ is analogous to the base step of whole the structural induction, considered above.

Consider the induction step $C = Dc$ for some set of bags $D$ and an individual bag $c$.
The induction hypothesis is that for any $r,s$ such that $r \subseteq s \subseteq b$, 
\[
\f(r,s,b,D)=\nmax{r}{s}{b}{\subtwo{b}{D}}.
\]
We obtain
\begin{align*}
\f(r,s,b,C) &= 
    \sum_{
        s' \cup  s'' = s \cap  c
    }\f(r,(s \setminus c)s',b, D) \cdot \g(r \cap c, s'',b,c)\\
    &=
    \sum_{
    s' \cup  s'' = s \cap  c
    }
    \nmax{r}{(s \setminus c)   s'}{b}{  \subtwo{b}{D}}\cdot \nmax{r\cap c}{s''}{b\cap c}{\sub{c}}\\
   &= |\bigcup_{s' \cup  s'' = s \cap  c} \smax{r}{(s \setminus c)   s'}{b}{  \subtwo{b}{D}} \times \smax{r\cap c}{s''}{b\cap c}{\sub{c}}|\\
    &=\nmax{r}{s}{b}{  \subtwo{b}{C}}
    \end{align*}
    where the second last equality follows by the fact that $\smax{r}{s}{d}{e}$ and $\smax{r}{s'}{d}{e}$ are disjoint for distinct $s$ and $s'$, and the
    last equality is due to the following claim. 
    \begin{claim}
       The mapping $\tau: r' \mapsto (r'_0,r'_1)$, where
       \begin{itemize}
           \item $r'_0\coloneqq r' \cap  \usubtwo{b}{D}$ and
           \item $r'_1 \coloneqq r' \cap \usub{c}$,
       \end{itemize}
       is a bijection from $\smax{r}{s}{b}{   \subtwo{b}{C}}$ into
       \begin{equation}\label{eq:set}
\bigcup_{
       \substack{s' \cup  s'' = s \cap  c\\ r \cap c\subseteq s' \cap s''}
} \smax{r}{(s \setminus c)   s'}{b}{  \subtwo{b}{D}} \times \smax{r\cap c}{s''}{b\cap c}{\sub{c}}.
       \end{equation}
    \end{claim}
    \begin{proof}    
    Clearly, $\tau$ is injective.
    We first show that $\tau(r)$ belongs to the set in \cref{eq:set}.  
    Let $r' \in \smax{r}{s}{b}{   \subtwo{b}{C}}$.     By definition (\cref{eq:eset}) we obtain
     $r'   \subseteq \usubtwo{b}{C}\setminus b$, 
    \begin{enumerate}[label=(\alph*)]
        \item\label{eq:max1} $ \maxrep{rr'}{\usubtwo{b}{C}}=s \cup  (\usub{C}\setminus b)$, and
        \item\label{eq:query} \(
    rr' \cap b_0 \not\models q
     \)
    for all $b_0 \in  \subtwo{b}{C}$.
    \end{enumerate}
    We need to show that
    \begin{enumerate}[label=(\roman*)]
             \item\label{it:second-part} $r'_0\in\smax{r}{(s \setminus c)   s'}{b}{  \subtwo{b}{D}}$, and
        \item\label{it:first-part} $r'_1\in
         \smax{r\cap c}{s''}{b\cap c}{\sub{c}}$
    \end{enumerate}
     for some $s$ and $s'$ such that $s'\cup s'' = s \cap c$ and $r \cap c\subseteq s' \cap s''$. For this, we let
     \begin{align}
         s' \coloneqq &     \maxrep{rr' \cap \usubtwo{b}{D}}{\usubtwo{b}{D}} \cap (b \cap c)\text{, and} \label{eq:s'}\\
         s''\coloneqq &\maxrep{rr' \cap \usub{c}}{\usub{c}} \cap (b \cap c).\label{eq:s''}
     \end{align}
The inclusion $r \cap c\subseteq s' \cap s''$ is clear.
Consider now the subtree consisting of $b$ and the union of subtrees rooted at bags in $C$. We can partition this subtree into $\sub{c}$ and $\subtwo{b}{D}$. Then, \cref{lem:sub} applied to $\db\coloneqq   \usubtwo{b}{C}$ $\db_1 \coloneqq \usub{c}$ and $\db_2 \coloneqq   \usubtwo{b}{D}$ yields
\begin{equation}\label{eq:max2}
 \maxrep{rr'}{   \usubtwo{b}{C}}
    = 
    \maxrep{rr' \cap \usub{c}}{\usub{c}}  %
    \cup
    \maxrep{rr' \cap  \usubtwo{b}{D}}{  \usubtwo{b}{D}}.  %
    \end{equation}
In particular, \cref{eq:max1,eq:max2} entails
\[
s \cap c = s \cap (b\cap c) =  \maxrep{rr'}{\usubtwo{b}{C}} \cap (b\cap c) = s' \cup s''.
\]
What remains to be shown are  \cref{it:first-part,it:second-part}.
     \paragraph{Proof of Item \ref{it:first-part}}
     Clearly $r'_1 \subseteq \usub{c}\setminus (b \cap c)$.
     By definition (\cref{eq:eset}), it then suffices to prove that 
         \begin{enumerate}[label=(\alph*1)]
        \item\label{eq:max1-1} $\maxrep{(r\cap c)r_1'}{\usub{c}}=s'' \cup  (\usub{c}\setminus b)$, and
         \item\label{eq:query-1} \(
    rr_1' \cap b_0 \not\models q
     \)
    for all $b_0 \in  \sub{c}$.
    \end{enumerate}
     Clearly, \cref{eq:query-1} follows by \cref{eq:query}. For \cref{eq:max1-1}, 
    we obtain 
    \begin{align}
    \maxrep{rr' \cap \usub{c}}{\usub{c}} \setminus (b \cap c) &= (\maxrep{rr'}{   \usubtwo{b}{C}}\cap \usub{c})\setminus (b\cap c)\nonumber
    \\
    &=\overbrace{((s \cap \usub{c})\setminus (b\cap c))}^{=\emptyset} \cup  (\usub{c}\setminus b)\nonumber\\
    &= \usub{c} \setminus (b \cap c).\nonumber %
    \end{align}
    The first equality follows by \cref{eq:max2} due to the connectivity condition of the tree decomposition, which entails that the intersection of $\usub{c}$ and $\usubtwo{b}{D}$ can contain only elements from $b \cap c$.
    The second equality stems from \cref{eq:max1}. The third equality is obtained by observing that any element in both $b$ and $\usub{c}$ belongs to $b\cap c$ due to the connectivity requirement.
    By \cref{eq:s''}, we can thus write
      $\maxrep{rr' \cap \usub{c}}{\usub{c}} = s''   (\usub{c} \setminus (b \cap c))$.
      Applying furthermore the connectivity condition, we can rewrite $rr' \cap \usub{c}$ as $(r\cap c)(r'\cap\usub{c})$; in particular, $r\cap c=r\cap \usub{c}$  since $r \subseteq b$, and any element in both $b$ and $\usub{c}$ belongs to $c$. Since $r'_1=r'\cap\usub{c}$,
      we obtain $\maxrep{(r\cap c)r'_1}{\usub{c}} = s''   (\usub{c} \setminus b )$, which is the statement of \cref{eq:max1-1}.

      \paragraph{Proof of Item \ref{it:second-part}}
           We have $r'_0 \subseteq \usubtwo{b}{D}\setminus b$, and
       now it suffices to prove that 
         \begin{enumerate}[label=(\alph*0)]
        \item\label{eq:max1-2} $\maxrep{rr_0'}{\usubtwo{b}{D}}=(s\setminus c)s'   (\usub{D}\setminus b)$, and
         \item\label{eq:query-2} \(
    rr_0' \cap b_0 \not\models q
     \)
    for all $b_0 \in  \subtwo{b}{D}$.
    \end{enumerate}
     Again, \cref{eq:query-2} follows by \cref{eq:query}.
Also, by the connectivity condition, \cref{eq:max1}, and \cref{eq:max2}, we obtain analogously to the previous case that
    \begin{align}
    \maxrep{rr' \cap \usubtwo{b}{D}}{\usubtwo{b}{D}} \setminus (b \cap c) &= (\maxrep{rr'}{   \usubtwo{b}{C}}\cap \usubtwo{b}{D})\setminus (b\cap c)\nonumber
    \\
    &=((s \cap \usubtwo{b}{D})\setminus (b\cap c)) \cup  (\usub{D}\setminus b)\nonumber\\
    &= (s \setminus c) \cup (\usub{D}\setminus b).\nonumber %
    \end{align}
 By \cref{eq:s'}, we can thus write
     $\maxrep{rr' \cap  \usubtwo{b}{D}}{ \usubtwo{b}{D}} = s'   (s \setminus  c )  (  \usub{D} \setminus  b)$. 
    Since $rr'\cap \usubtwo{b}{D}=r(r'\cap \usubtwo{b}{D})$ and $r'_0\coloneqq r' \cap  \usubtwo{b}{D}$,
     we obtain $\maxrep{rr'_0}{ \usubtwo{b}{D}}=(s \setminus  c )s'     (  \usub{D} \setminus  b)$, which is the statement of \cref{eq:max1-2}.

    Having concluded that \cref{it:second-part,it:first-part} hold, we have now proven that $\tau(r)$ belongs to the set in \cref{eq:set}.
 It remains to show that $\tau$ is surjective. Let $r_0\in \smax{r}{(s \setminus c)s'}{b}{\usubtwo{b}{D}} $ and $r_1\in \smax{r\cap c}{s''}{b\cap c}{\usub{c}}$ for some $s',s''$ such that $s' \cup  s'' = s \cap  c$ and $r \cap c\subseteq s' \cap s''$. We need to find $r' \in \smax{r}{s}{b}{\usubtwo{b}{C}} $ such that $r_0=r' \cap \usubtwo{b}{D}$ and $r_1 = r' \cap \usub{c}$.
 Now, we simply let $r'\coloneqq r_0 r_1$. It is easy to see that the  above two equality conditions hold. To show that $r'$ belongs to $ \smax{r}{s}{b}{ \usubtwo{b}{C}}$, we need to prove that
 $r' \subseteq \usubtwo{b}{C}\setminus b$,
\begin{enumerate}[label=(\Alph*)]
    \item\label{it:max-A} $\maxrep{rr'}{\usubtwo{b}{C}}=s(\usubtwo{b}{C}\setminus b)$, and
    \item\label{eq:query-B} \(
    rr' \cap b_0 \not\models q
     \)
    for all $b_0 \in  \subtwo{b}{C}$.\end{enumerate}
 By definition (\cref{eq:eset}),
 $r_0 \subseteq \usubtwo{b}{D} \setminus b$ is such that 
 \begin{enumerate}[label=(\Alph*0)]
\item\label{it:max-A0} $\maxrep{rr_0}{\usubtwo{b}{D}}=(s\setminus c)s'( \usubtwo{b}{D} \setminus b)$, and
    \item\label{eq:query-B0} \(
    rr_0 \cap b_0 \not\models q
     \)
    for all $b_0 \in  \subtwo{b}{D}$.
 \end{enumerate}
 Moreover, $r_1\subseteq  \usub{c} \setminus (b \cap c)$ is such that
  \begin{enumerate}[label=(\Alph*1)]
\item\label{it:max-A1} $\maxrep{(r\cap c)r_1}{\usub{c}}=s''(\usub{c} \setminus (b \cap c)),$ and
  \item\label{eq:query-B1} \(
    rr_1 \cap b_0 \not\models q
     \)
    for all $b_0 \in  \sub{c}$. 
 \end{enumerate}
We observe first that $r'=r_0r_1 \subseteq \usubtwo{b}{C}\setminus b$. Furthermore, \cref{eq:query-B} is a consequence of \cref{eq:query-B0,eq:query-B1}.
For \cref{it:max-A}, consider the statement of \cref{lem:sub}, 
letting $\db\coloneqq \usubtwo{b}{C}$, $\db_1\coloneqq \usubtwo{b}{D}$, $\db_2\coloneqq \usub{c}$, and $\rep \coloneqq r r'$. Then $\rep \cap \db_1 = rr_0$ and $\rep \cap \db_2 = (r\cap c)r_1$,
and using \cref{it:max-A0,it:max-A1} the lemma yields       
\begin{align*}
\maxrep{rr'}{\usubtwo{b}{C}}&=
\maxrep{\rep}{\db} %
= \maxrep{\rep\cap \db_1}{\db_1} \cup \maxrep{\rep \cap \db_2}{\db_2}, %
\\
&=(s\setminus c)s'( \usubtwo{b}{D} \setminus b)
\cup 
s''(\usub{c} \setminus (b \cap c))\\
&=s's''(s\setminus c)(\usubtwo{b}{C}\setminus b )
=s(\usubtwo{b}{C}\setminus b).
\end{align*}
Thus \cref{it:max-A} follows. We conclude that $r' \in \smax{r}{s}{b}{\usubtwo{b}{C}}$, as required.
This shows that $\tau$ is surjective, concluding the proof of the claim.
    \end{proof}
This concludes the induction proof for the lemma.
\end{proof}

\section{Gaifman Graphs and MSO-encoding}\label{sect:appendixB}

\subsection{Functional Dependencies}\label{sect:appendixFD}

\different*
\begin{proof}
    For simplicity, we consider the empty constraint set $\Sigma$.
    This is not a restriction as the construction can easily be adapted to the case of non-empty $\Sigma$.
    
    Let $\db_{n} \coloneqq \{R(1), \dots ,R(n),S(1), \dots S(n)\}$ and $q\coloneqq \exists x y(R( x)\wedge S( y))$, and let us write $\calI_n=(\db_n, \emptyset, q)$.
    Then $\hyperGraphRepr{\calI_n}$  consists of solution-edges $\{R(i),S(j)\}$, for $i,j\in [n]$, 
    meaning that it is the complete bipartite graph $K_{n,n}$ which has treewidth $n$.
    However, the Gaifman graph of ${\calA_{\calI_n}}$ is an edgeless graph as no atoms in $q$ share variables.
    Consequently, the Gaifman graph is trivially a tree, meaning its treewidth is $1$.
    This concludes the proof of the first item.
    
    For the second item, let $\db'_{n} = \{R(1,*), \dots ,R(n,*), S(*,1), \dots ,S(*,n),T(-1), \dots ,T(-n)\}$ and
     $q' = \exists xyz (R(x,y)\land S(y,z)\land T(z))$, and let us write $\calI'_n=(\db'_n, \emptyset, q')$. %
    This time, $\hyperGraphRepr{\calI'_n}$ contains no solution-edges (and vacuously no conflict-edges), resulting in treewidth $1$.
    In contrast, the Gaifman graph of ${\calA_{\calI'_n}}$ contains edges due to the common variables in the atoms $R(x,y)$ and $S(y,z)$. 
    Specifically, this graph is the complete bipartite graph $K_{n,n}$ %
    which has treewidth $n$.
\end{proof}

\correct*

\begin{proof}
Suppose $q$ is of the form $\bigvee_{\ell \leq m} q_\ell$.

    (``$\Longleftarrow$'')
    Suppose there is repair $\rep\subseteq \db $ %
    that does not satisfy $q$. 
    We prove that $\calA_\calI\not\models \Phi_\calI$.
    We find an interpretation $\Theta$ for the variable $T$ corresponding to facts in $\rep$, failing $\Phi_\calI$. %
    That is, we let $\Theta(T)\coloneq \rep$.
    
    (I). $(\calA_\calI,\Theta)\models \varphi_{\text{repair}}(T)$. 
    Since $\rep$ satisfies every functional dependency $f\in \Sigma$ over each $R\in\sch$, we have $(\calA_\calI,\Theta)\models \varphi_{\text{sat}}(T)$.
    To see this, we note that $\{R(\vec a),R(\vec b)\}\models f$ for each $R(\vec a),R(\vec b)\in \rep$ and {\FD} $f\in \Sigma$ over $R$.
    Then, $\varphi_{\text{repair}}(T)$ is satisfied in $\calA_\calI$ by the  interpretation $\Theta$ as $(R(\vec a), R(\vec b))\not\in\depfails^{\calA_\calI}$ for any $R(\vec a), R(\vec b) \in \Theta(T)$ and {\FD} $f$ over $R$. 
    As a result, $(\calA_\calI,\Theta)\models \varphi_{\text{sat}}(T)$. 
    Finally, $(\calA_\calI,\Theta)\models \varphi_{\text{repair}}(T)$ since $\rep$ is a maximal subset of $\db$ satisfying $\Sigma$. 

    (II). Next, we prove that $(\calA_\calI,\Theta) \not\models \varphi_{\text{$T,q_\ell$-sat}}$ for each $\ell \leq m$.
    Since $\rep$ is a repair that does not satisfy $q$, no sequence %
    of facts $\{q_1(\vec a_1),\dots, q_{s_\ell}(\vec a_{s_\ell})\}\subseteq \rep$ forms a solution to $q_\ell$.
    Recall $q_\ell$ is of the form $\exists \vec x (\bigwedge_{i \leq s_\ell} q_i(\vec t_i) \land \bigwedge_{s_\ell< i \leq s'_\ell} t_i \neq u_i)$, where the first conjunction consists of relational atoms
    and the second of inequality atoms.
    Assume toward contradiction that $\calA_\calI\models \varphi_{\text{$T,q_\ell$-sat}}$. Then we find facts $f_1,\dots, f_{s_\ell}$ from $T$ such that
    $\calA_\calI \models \bigwedge_{(i,j)\in L}  \linked_{i,j}(f_i,f_j) $. In other words, two facts $f_i$ and $f_j$, for $1 \leq i,j \leq s_\ell$, are $q_\ell$-consistent whenever $q_i(\vec t_i),q_j(\vec t_j)$ are $q_\ell$-linked.
    We prove by induction on $k\leq s_\ell$ that there exists an assignment $h$ from the variables to the constants such that
    $q_i(h(\vec t_i))=f_i$, for $i \leq k-1$, and $h(t_i)\neq h(u_i)$, for $s< i \leq s'$.
    This entails that    $%
    \rep \models \exists \vec x(\bigwedge_{i \leq \ell} q_i(\vec t_i) \land \bigwedge_{s< i \leq s'} t_i \neq u_i )$. Thus we obtain $\rep \models q$, which leads to a contradiction. Consequently, $\calA_\calI\not \models \varphi_{\text{$T,q_\ell$-sat}}$, and therefore it suffices to prove the induction claim.
    
    The base case $k=1$. Since $(1,1)\in L$ vacuously, we have $\calA_\calI \models \linked_{1,1}(f_1,f_1)$. Thus $f_1$ is $q_\ell$-consistent with itself, and by definition there is an assignment $h$ such that $q_1(h(\vec t_1))=f_1$, and $h(t_i)\neq h(u_i)$, for $s< i \leq s'$.
    
    The base case $k>1$. %
    The induction hypothesis states that there is an assignment $h$
    such that $q_i(h(\vec t_i))=f_i$, for $i \leq k-1$, and $h(t_i)\neq h(u_i)$, for $s< i \leq s'$. Consider then the fact $f_k$, and suppose it is of the form
    $R(a_1, \dots ,a_p)$. Suppose also $\vec t_k$ is of the form $(t_1, \dots ,t_p)$. Consider a term $t_j$, for $j \leq p$. If $t_j$ is a constant, then $t_j=a_j$ due to the fact that $f_k$ is $q_\ell$-consistent with itself. Suppose then $t_j$ is a variable $x$. If $x$ appears in some sequence $\vec t_q$, $q \leq k-1$, then $t_j=h(x)$ due to the fact that $f_k$ and $f_q$ are $q_\ell$-consistent. Otherwise, if $x$ appears in none of the listed sequences, we change $h$ by setting $h(x)=t_j$. Note that if $t_{j'}$ is also the variable $x$, for $j \neq j'$, then $t_j=t_{j'}$ because of the $q_\ell$-consistency of $f_k$ with itself. Thus we observe $q_k(h(\vec t_k))=f_k$. Similarly, one can observe using $q_\ell$-consistency that it remains to be the case that $h(t_i)\neq h(u_i)$, for $s< i \leq s'$.
    This concludes the induction proof. 

    (``$\Longrightarrow$'')
    Suppose $\calA_\calI\not \models \Phi_\calI$. 
    Then, there exists an interpretation $\Theta$ to the set variable $T$ such that $(\calA_\calI,\Theta)\not\models \Phi'$, where $\Phi'$ is a formula in the vocabulary $\tau_\calI\cup \{T\}$ obtained from $\Phi_\calI$ by removing the MSO-quantifier in the prefix.
    We construct a repair $\rep$ of $\db$ that does not satisfy $q$ from the assignment $\Theta$ of the set variable $T$.
    That is, we let $\rep = \Theta(T)$.
    Clearly, $\rep$ is a repair
    due to $(\calA_\calI,\Theta)\models \varphi_{\text{repair}}(T)$. %
    Thus it remains to prove that $\rep\not\models q$. 
    Suppose to the contrary that $\rep\models q$. Then $\rep\models q_\ell$, for some $\ell\leq m$. Suppose $q_\ell$ is of the form $\exists \vec x (\bigwedge_{i \leq s_\ell} q_i(\vec t_i) \land \bigwedge_{s_\ell< i \leq s'_\ell} t_i \neq u_i)$, where the first conjunction consists of relational atoms
    and the second of inequality atoms. Let $h$ be the assignment of the variables in $\vec x$ witnessing $\rep\models q_\ell$. We obtain $q_i(h(\vec t_i))\in \rep$, for $i \leq s_\ell$ and $h(t_i)\neq h(u_i)$, for $s_\ell< i \leq s'_\ell$. Next, let $\alpha$ be the assignment that maps the variables $x_1,\dots,x_s$ respectively to the facts $q_1(h(\vec t_1)), \dots ,q_{s_\ell}(h(\vec t_{s_\ell}))$. Clearly, $(\calA_\calI,\Theta) \models_\alpha \bigwedge_{i\leq s} T(x_i) \land\bigwedge_{(i,j)\in L}  \linked_{i,j}(x_i,x_j) $, meaning that we obtain  $(\calA_\calI,\Theta) \models\varphi_{\text{$T,q_\ell$-sat}}$. This, however, contradicts the fact that $(\calA_\calI,\Theta)\not\models \Phi'$. By the contradiction we obtain $\rep\not\models q$, as required.
\end{proof}

\lemfamily*

\begin{proof}
    Let $\calI=(\db, \Sigma, q)$ be an instance, where $q$ is of the form $\bigvee_{\ell\leq m} q_\ell$.
    The function $f$ returns the formula $\Phi_\calI$ for each instance $\calI$.
    It remains to prove that the function $f$ is $\FPT$-computable and $\kappa$-bounded.
    To achieve this, it suffices to prove that the size of $f(\calI)= \Phi_\calI$ is bounded by the parameter values of the instance $\calI$.
    We have the following observations regarding the size of each subformula.
\begin{itemize}
    \item The subformulas %
    $\varphi_{\text{sat}}$, $\varphi_{\tilde T \supsetneq T}$, and $\varphi_{\text{repair}}$ are of constant size.
    \item The formulas $\varphi_{\text{$T,q_{\ell}$-sat}}$, for $\ell \leq m$, have size at most quadratic in the number of atoms in $q$ (due to $L$).
    \item Finally, the size of $\Phi_\calI$ depends cubically on the number of atoms in $q$ (due to $m \leq \atomsq(q)$).
\end{itemize}
Considering $\atomsq(q)+\twG$ as the parameter, $|\Phi_\calI|$ is bounded by a function in the parameter (in particular, $\atomsq(q)^3$).
Therefore the function $f$ is $\FPT$-computable and $\kappa$-bounded. This proves the lemma.
\end{proof}

\subsection{Denial Constraints}\label{app:dc}

\correctDC*
\begin{proof}
    The proof is analogous to the proof of Theorem~\ref{thm:mso-correct-fds}.
    We sketch the main differences here and establish their correctness.
    
    ``$\Longleftarrow$''. 
    Suppose there exists a repair $\rep\subseteq \db$ that does not satisfy $q$.
    We prove that the interpretation $\Theta$ for the set variable $T$ corresponding to the facts in $\rep$ does not satisfy $\Psi_\calI$ in $\calB_\calI$. Thus we let $\Theta(T)\coloneq \rep$
    
    (I). $(\calB_\calI,\Theta)\models \varphi_{\text{repair}}(T)$. 
    Since $\rep$ satisfies  $ \Sigma$, we have $(\calB_\calI,\Theta)\models \varphi_{\text{sat}}(T)$.
    To see this, we note that %
    $C\not\subseteq \rep$ for any conflict $C$ in $\db$ due to a DC $\sigma\in \Sigma$. %
    Then, Formula~(1a) is satisfied by the given interpretation $\Theta$, as $(f_1,\dots,f_k)\not\in\depfails^\calB$ for $f_1,\dots ,f_n \in \Theta(T)$. %
    As a result, $(\calB_\calI,\Theta)\models \varphi_{\text{sat}}(T)$. 
    Finally, $(\calB_\calI,\Theta)\models \varphi_{\text{repair}}(T)$ since  $\rep$ constitutes a subset-maximal subset of $\db$ satisfying $\Sigma$. 
    Then, analogously to the proof of \cref{thm:mso-correct-fds}, we obtain $( \calB_\calI,\Theta) \not\models\bigvee_{\ell\leq m}  \varphi_{\text{$T,q_{\ell}$-sat}}$. 
    
    ``$\Longrightarrow$''. 
    Suppose $\calB_\calI\not\models \Psi_\calI$. 
    Then, there exists an interpretation $\Theta$ to the set variable $ T$ such that $(\calB_\calI,\Theta)\not\models \Psi'$, where $\Psi'$ is a formula in the vocabulary $\tau_\calI\cup \{T \}$ obtained from $\Psi_\calI$ by removing the MSO-quantifier in the prefix.
    Using the interpretation $\Theta$ of the set variable $T$, we construct a repair $\rep$ for $\db$ failing $q$.
    To this aim, we let $\rep =  \Theta(T)$.
    Now, we prove that $\rep$ 
    satisfies every DC in $\Sigma$.
    Suppose to the contrary, there exists a conflict $C\subseteq \rep$ due to some DC $\sigma\in \Sigma$. 
    Moreover, we can assume $|C|= k$. 
    Since $C\subseteq \rep$, there are $x_1,\dots,x_k$ such that $(\calB_\calI,\Theta)\models \bigwedge_{j\leq k} T(x_j) $ (due to the way $\rep$ is defined) as well as $(\calB_\calI,\Theta)\models \depfails(x_1,\dots,x_k)$.
    However, this contradicts Formula~(1a) which is true for all sequences of length $k$ in $\Theta(T)$. 
    As a result, $\rep \models \sigma$ for each $\sigma\in \Sigma$. Following this, we can analogously to the proof of \cref{thm:mso-correct-fds} show that $\rep \not\models q$. 
\end{proof}

\lemfamilydcs*
\begin{proof}
    Let $\calI=(\db, \Sigma, q)$ be an instance.
    The function $f$ returns the formula $\Psi_\calI$ for each instance $\calI$.
    It remains to prove that the function $f$ is $\FPT$-computable and $\kappa$-bounded.
    To achieve this, we prove that the size of $f(\calI)= \Psi_\calI$ is bounded by the parameter values of the instance $\calI$.
    The only difference between $\Psi_\calI$ and $\Phi_\calI$ is the subformula $\varphi_{\text{sat}}$.
    However, the size of the subformula $\varphi_{\text{sat}}$ in $\Psi_\calI$ depends linearly on the maximum number of atoms $k$ in any denial constraint of $\Sigma$.
    As a result, similar to the case of $\Phi_\calI$, $|\Psi_\calI|$ is also bounded by a function in $\atomsq(q)$, since $k$ is constant.
    Therefore the function $f$ is $\FPT$-computable and $\kappa$-bounded. This proves the lemma.
\end{proof}

\gaifmandc*

\begin{proof}
    Since the set $\Sigma$ of DCs is fixed, its size and maximum arity is a constant.
    For data (resp., combined) complexity, the query $q$ is fixed (part of the input).
    Consequently (using Lemma~\ref{lem:family2}), one can bound the size $|\Psi_\calI|$  by a function of the parameter $\atomsq(q)$.
    Given an instance $\calI$, we compute $\Psi_\calI$ in $\FPT$ time. 
    Then, the model checking problem of the instance $(\calA_\calI,\Psi_\calI)$ can be decided in $\FPT$ time in $\twG$ by Courcelle's theorem. 
    Considering $\atomsq+\twG$ as parameterisation, the size of the formula $\Psi_\calI$ is $\kappa$-bounded, whereby the theorem applies.
\end{proof}

\end{document}